\documentclass[12pt,a4paper]{article}

\usepackage{amsmath}
\usepackage{amssymb}
\usepackage{mathrsfs}
\usepackage{graphicx}
\usepackage{enumerate}
\usepackage[pdftex,colorlinks=true,linkcolor=blue,citecolor=blue,urlcolor=blue]{hyperref}

\numberwithin{equation}{section}

\newtheorem{theorem}{Theorem}[section]
\newtheorem{lemma}[theorem]{Lemma}
\newtheorem{proposition}[theorem]{Proposition}

\def\eq#1 { \begin{equation} #1 \end{equation} }
\def\eqn#1{ \begin{eqnarray} #1 \end{eqnarray} }
\def\nn { \nonumber }
\def\half{\frac{1}{2}}
\def\a{{\alpha}}
\def\b{{\beta}}
\def\l{{\lambda}}
\def\s{{\sigma}}
\def\D{\Delta}
\def\bg{{\overline{g}}}

\def\bmu{{\overline{\mu}}}
\def\bnu{{\overline{\nu}}}

\def\bx{{\overline{x}}}

\def\bb{{\overline{\beta}}}
\def\ba{{\overline{\alpha}}}

\def\bBox{{\overline{\Box}}}

\def\bX{\overline{X}}
\def\etab{{\overline{\eta}}}
\def\bg{{\overline{\gamma}}}
\def\bd{{\overline{\delta}}}

\def\cA{\mathcal{A}}
\def\cD{\mathcal{D}}
\def\cO{\mathcal{O}}
\def\cH{\mathcal{H}}
\def\vx{{\vec{x}}}
\def\vj{{\vec{j}}}

\def\vL{{\vec{L}}}
\def\vM{{\vec{M}}}
\def\Reals{\mathbb{R}}
\def\Complex{\mathbb{C}}
\def\testT{\mathscr{T}_{\rm T}}
\def\testTT{\mathscr{T}_{\rm TT}}

\def\grad{{\rm grad}}
\def\met{{\rm metric}}
\def\Re{{\rm Re}\,}

\def\ph{\phantom}
\def\P{\text{P}}
\def\Lie{\mathscr{L}}
\def\d{\partial}
\def\vol{{\rm vol}(S^D)}
\def\scrI{\mathscr{I}}
\def\WF{{\rm WF}}
\def\WFA{{\rm WF_A}}
\def\2F1#1#2#3#4{\,\phantom{}_2F_1\left[#1\,,\,#2\,;\,#3\,;\,#4\right] }
\def\C#1{\left\langle #1 \right\rangle}
\def\CO#1{\left\langle #1 \right\rangle_\Omega}
\def\ket#1{\left| #1 \right\rangle}
\def\bra#1{\left\langle #1 \right|}

\def\T2#1#2#3#4{(\nabla_{(#1}\nabla^{(#3}Z)(\nabla_{#2)}\nabla^{#4)}Z)}
\def\GG#1#2#3#4{g_{(#1}^{\ph{#1}(#3} g_{#2)}^{\ph{#2)}#4)}}

\author{
  Ian A. Morrison${{}^1}$\thanks{\href{mailto:i.morrison@damtp.cam.ac.uk}
    {i.morrison@damtp.cam.ac.uk}} \\ \\
  {\it ${}^1$DAMTP, Centre for Mathematical Sciences,}\\
   {\it University of Cambridge, U.K.}
}

\date{\today}

\def\Title#1{\title{\vspace{-0.65in}#1}}

\Title{On cosmic hair and ``de Sitter breaking'' in linearized 
  quantum gravity}

\begin{document}

\maketitle

\begin{abstract}
  We quantize linearized Einstein-Hilbert gravity on de Sitter 
  backgrounds in a covariant gauge. We verify the existence of a 
  maximally-symmetric (i.e.~de Sitter-invariant) Hadamard state 
  $\Omega$ for all globally hyperbolic de Sitter backgrounds in all 
  spacetime dimensions $D \ge 4$ by constructing the state's 2-point 
  function in closed form. 
  This 2-pt function is explicitly maximally symmetric.
  We prove an analogue of the Reeh-Schlieder theorem for linearized
  gravity.
  Using these results we prove a cosmic no-hair theorem for linearized 
  gravitons: for any state in the Hilbert space constructed from $\Omega$, 
  the late-time behavior of local observable correlation functions reduces 
  to those of $\Omega$ at an exponential rate with respect to proper time.
  We also provide the explicitly maximally-symmetric graviton 2-pt 
  functions in a class of generalized de Donder gauges 
  suitable for use in non-linear perturbation theory.
  Along the way we clarify a few technical aspects
  which led previous authors to conclude that these 2-pt functions 
  do not exist.
\end{abstract}

\newpage

\tableofcontents

\section{Introduction and summary}
\label{sec:intro}

Perturbative quantum gravity on de Sitter (dS) backgrounds is of
considerable interest. de Sitter spacetime is the maximally-symmetric 
model of an expanding cosmology and may be used to approximate
our universe both in the inflationary era as well as our current 
epoch of accelerated expansion. 
At the level of phenomenology, perturbative gravitons on de Sitter 
may be used to study tensor fluctuations in the cosmic microwave 
background (CMB) \cite{Weinberg:2008zzc}.
Of more theoretical interest, previous studies have led many 
authors to conjecture interesting infrared (IR) effects associated to de
Sitter gravitons. 
It has been argued that
even in perturbation theory one may find
large quantum (i.e., loop) IR corrections
\cite{Tsamis:1994ca,Giddings:2010nc,Giddings:2011zd,Giddings:2011ze,Kitamoto:2012ep}, 
the dynamical screening of the effective cosmological constant
\cite{Tsamis:1996qm,Tsamis:1996qq,Mottola:2010gp,Tsamis:2011uq},
and other destabilizing phenomena 
\cite{Mottola:1985qt,Mazur:1986et,Floratos:1987aa}.
An interesting feature of these effects is that they appear to
break de Sitter symmetry. This has caused some controversy as
the precise mechanism by which de Sitter symmetry is broken has 
not been adequately explained. To this day both the existence and role
of a maximally-symmetric (a.k.a.~de Sitter-invariant) graviton state
is heatedly debated.

It is a general expectation of quantum field theory (QFT) in 
curved spacetime \cite{Hollands:2008vx} that when the background
spacetime admits isometries the associated isometry group should
play a crucial role in organizing the theory, much like how the Poincar\'e
group organizes QFT in Minkowski spacetime.
In the context of a de Sitter background, this expectation is realized at
least for QFTs with a mass gap. For instance, it is known that 
interacting massive scalar QFTs enjoy a unique maximally-symmetric state
\cite{Marolf:2010nz,Hollands:2010pr,Higuchi:2010aa,Korai:2012fi}.
This state may be interpreted as the interacting ``Bunch-Davies vacuum''
in the cosmological chart or as the thermal state at the 
de Sitter temperature in the static chart.
Most important is the role this maximally-symmetric state plays in the 
``cosmic no-hair theorem'' of \cite{Marolf:2010nz,Hollands:2010pr}
which states that
the expectation values of local observables in a \emph{generic} state
limit to those of the maximally-symmetric state in the asymptotic regions. 
This makes the maximally-symmetric state the most 
relevant state when studying the asymptotic behavior 
of expectation values of the theory.

It is of great interest to determine under what circumstances
a similar cosmic no-hair theorem exists for perturbative quantum gravity.
In classical gravity, asymptotically de Sitter 
spaces are known to be non-linearly stable 
under an appropriate class of perturbations 
\cite{Friedrich:1986aa,Anderson:2004ir}.
Moreover, a no-hair theorem of Wald \cite{Wald:1983ky} states that
initially expanding homogeneous solutions to vacuum Einstein's equations with 
positive cosmological constant 
exponentially evolve toward locally de Sitter spaces.
These classical results suggest that at least at tree-level quantum 
gravity should admit a cosmic no-hair theorem very similar to that of 
the dS QFTs described above. In particular, these results imply
the existence of a maximally-symmetric graviton state which acts 
as the attractor state for local observables at late 
times. This is further supported by the recent non-perturbative semi-classical
analysis of \cite{Frob:2013ht}.

The goal of this paper is to make precise the cosmic no-hair theorem
for \emph{linearized} quantum gravity on de Sitter.
We quantize linearized Einstein-Hilbert gravity in $D \ge 4$
spacetime dimensions on a de Sitter background.
We employ an algebraic approach to quantization which allows us to 
quantize the theory in covariant gauges, greatly 
facilitate our computations, and make our results transparent. 
The algebraic approach also allows us to simultaneously discuss 
the quantization on various de Sitter charts which may be be taken to 
define globally hyperbolic spacetimes in their own right.
Excepting the Appendix all computations are performed in Lorentz 
signature, though in reality most of our manipulations are insensitive 
to the metric signature.
We have four main results which we summarize now:
\begin{enumerate}[R1.]
  \item \label{res:TT}
    We verify the existence of a maximally-symmetric state $\Omega$
    by constructing the graviton 2-pt correlation function 
    of this state explicitly in two classes of gauges.
    This 2-pt function may be written in closed form in terms of 
    maximally-symmetric bi-tensors and is thus manifestly maximally symmetric.
    We verify that the state $\Omega$ is Hadamard and satisfies the 
    positivity (a.k.a.~``unitarity'') condition.
  \item \label{res:RS}
    We prove an analogue of the Reeh-Schlieder theorem for linearized
    gravity (Theorem~\ref{thm:RS}). 
    The Reeh-Schlieder theorem states that the set of states
    generated from $\Omega$ by the algebra of local observables of a 
    contractible region is dense on the Hilbert space $\cH_\Omega$
    constructed from $\Omega$ via the GNS construction.
    The Reeh-Schlieder property is a remarkable attribute of 
    \emph{local} quantum field theories and in the context of
    gravity is clearly special to the linear theory.
    Nevertheless, this theorem allows us to prove
    the next result with some rigour.
  \item \label{res:hair}
    We prove a cosmic no-hair theorem for linearized gravity
    (Theorem~\ref{thm:hair}) which states the following:
    let $\Psi$ be a state on the Hilbert space $\cH_\Omega$ and let 
    $A$ be an observable whose support
    is compact and contractible. At sufficiently late times the 
    expectation value $\C{A(\tau)}_\Psi$ approaches $\CO{A(\tau)}$ rapidly:
    \eq{
      \left|\C{ A(\tau) }_\Psi - \CO{ A(\tau) } \right| 
      < c\, e^{-2 \tau} .
    }
    Here $c$ is a finite, non-negative constant and $\tau$ is 
    the proper time separation from any reference point in the past.
    This result follows rather simply from the fact that the 
    expectation values of local observables in the state
    $\Omega$ obey cluster decomposition at large timelike and achronal 
    separations. 
\end{enumerate}
This no-hair theorem may be regarded as a quantum formulation, for 
linearized perturbations, of the classical stability theorems of 
Freidrich \cite{Friedrich:1986aa} and Anderson \cite{Anderson:2004ir}.
We emphasize that this result holds in linearized quantum gravity,
i.e.~we do not consider self-interactions or coupling to matter fields.
As such the no-hair theorem does not directly constrain the remarkable 
effects mentioned in the first paragraph, though we believe it provides
a valuable perspective. Our results \emph{do} show that ``de Sitter breaking'' 
\cite{Woodard:2004ut,Miao:2010vs} is not a phenomenon of linearized 
quantum gravity.

Our last result, while still a result of linearized quantum gravity,
will be useful mostly in the context of non-linear perturbation theory:
\begin{enumerate}[R1.]\setcounter{enumi}{3}
  \item\label{res:dD}
    We compute the graviton 2-pt function of $\Omega$ in the
    one-parameter class of generalized de Donder gauges which satisfy
    the gauge condition
    \eq{
      \nabla^\nu h_{\mu\nu}(x) 
      - \frac{\beta}{2} \nabla_\mu h^\nu_{\ph{\nu}\nu}(x) = 0 ,
      \quad \beta \in \Reals .
    }
    This class of gauges is of interest because the gauge condition
    is generally covariant and may be imposed in non-linear 
    perturbation theory (unlike transverse traceless gauge).
    We obtain a manifestly maximally-symmetric expression for the
    2-pt function for all but a discrete set of values for the gauge
    parameter $\beta$. For the case $D=4$ our result agrees with the 
    expression that may be obtained from a limit of the 
    ``covariant gauge'' 2-pt function of 
    \cite{Higuchi:2001uv,Higuchi:2000ge,Faizal:2011iv}.
\end{enumerate}    
These de Donder 2-pt functions are not used to obtain
(\href{res:RS}{R2}) or (\href{res:hair}{R3}).

There is a large literature debating the existence of a manifestly 
maximally-symmetric graviton 2-pt function on de Sitter backgrounds -- see e.g.
\cite{Higuchi:1991tk,Kleppe:1993fz,Woodard:2004ut,Miao:2010vs,Urakawa:2010it,Higuchi:2011aa,Faizal:2011iv} and references therein.
Our results (\href{res:TT}{R1}) and (\href{res:dD}{R4}) settle
this debate at least within the context of our quantization scheme.
We note in particular two points of contact with the existing literature.
First, although our 2-pt functions are computed in Lorentz signature
they agree with the analytic continuation of Euclidean 2-pt functions 
constructed on the Euclidean sphere $S^D$. 
Several previous authors have utilized this technique
\cite{Allen:1986dd,Floratos:1987aa,Higuchi:2000ge,Higuchi:2001uv,Park:2008ki},
but its validity has been debated \cite{Miao:2009hb,Miao:2010vs}. 
Since the health of our results has been verified in Lorentz 
signature, we conclude that at least 
for the cases we consider no pathologies arise from this analytic
continuation process.

Second, our results (\href{res:TT}{R1}) and (\href{res:dD}{R4}) 
are technically in conflict with the claims of 
\cite{Miao:2011fc} and \cite{Mora:2012zi}: these works claim that
there do not exist maximally-symmetric solutions to the
graviton 2-pt Schwinger-Dyson equations in transverse traceless
or generalized de Donder gauges.
This conflict stems from the fact that these gauge conditions do not
completely fix the gauge freedom. Refs.~\cite{Miao:2011fc,Mora:2012zi} 
do not explore the full space of solutions consistent with these 
partial gauge conditions; instead, implicit in their analysis is an
additional boundary condition
(imposed at spacelike infinity in the Poincar\'e chart)
which is incompatible with the maximally symmetric solutions.
We discuss this in further detail in \S\ref{sec:compare}.
This conflict aside, we emphasize that the
\emph{state} defined by the less-symmetric 2-pt functions of 
\cite{Miao:2011fc,Mora:2012zi}
is equivalent to the maximally-symmetric state $\Omega$ as probed
by all local observables. For certain classes of observables this
fact has been pointed out before
\cite{Urakawa:2010it,Higuchi:2011aa,Faizal:2011iv,Higuchi:2012vy}.

The remainder of this paper is organized as follows. We begin in 
\S\ref{sec:prelims} by introducing preliminary material needed for our study.
This includes brief reviews of de Sitter spacetime (\S\ref{sec:dS}) and 
classical linearized gravity (\S\ref{sec:gravity}). Our quantization
scheme is described in detail in \S\ref{sec:quantum}.
In \S\ref{sec:2pt} we construct the 2-pt function of $\Omega$ in 
transverse traceless gauge. This construction is straight-forward 
but utilizes a great amount of simple technology which takes some
time to describe. The 2-pt function is finally
computed in \S\ref{sec:TT}. In \S\ref{sec:compare} we compare our 
findings with previous results in the literature.
The two theorems (\href{res:RS}{R2}) and (\href{res:hair}{R3}) are presented 
in \S\ref{sec:consequences}.
Finally, in \S\ref{sec:deDonder} we compute the 2-pt function of $\Omega$
in generalized de Donder gauge. We once again compare our results
to those in the literature in \S\ref{sec:compare2}. \\

\noindent{\it Note Added.}\hspace{2pt} Since this paper was first 
posted to the pre-print arXiv a lively critique of this work has 
appeared \cite{Miao:2013isa}.

\section{Preliminaries}
\label{sec:prelims}

\subsection{de Sitter space}
\label{sec:dS}

In this paper we consider Einstein-Hilbert gravity in $D$ spacetime 
dimensions with positive cosmological constant $\Lambda > 0$. The
classical theory may be defined by the action
\eq{
  S_{\rm EH} = \frac{1}{16\pi G} \int d^Dx \sqrt{-g(x)}
  \left( R(g) - 2\Lambda \right) ,
}
or equivalently by the equations of motion
\eq{ \label{eq:EinsteinEq}
  G_{\mu\nu}(g) + \Lambda g_{\mu\nu}(x) = 0 ,
}
where $R(g)$ and $G_{\mu\nu}(g)$ are the Ricci scalar and Einstein tensor
constructed from the metric $g_{\mu\nu}(x)$ respectively.
de Sitter space is the maximally symmetric solution to these equations. 
The $D$-dimensional de Sitter manifold $dS_D$ may be defined as the 
single-sheet hyperboloid in an $\Reals^{D,1}$ embedding space:
\eq{
  dS_D :=
  \left\{ X \in \Reals^{D,1}\; | \; X \cdot X = \ell^2 \right\} .
}
The de Sitter radius $\ell$ is related to the cosmological constant
via
\eq{
  \Lambda = \frac{(D-1)(D-2)}{2\ell^2} .
}
The full de Sitter manifold has the topology $\Reals\times S^{D-1}$ where
$\Reals$ is the timelike direction; it has two conformal boundaries 
$\scrI^\pm$ which are Euclidean spheres $S^{D-1}$.
From the embedding space description it is manifest that the isometry 
group of de Sitter, a.k.a.~the de Sitter group, is $SO(D,1)$.

\begin{figure}[t!]
  \centering
  \includegraphics[width=0.3\textwidth]{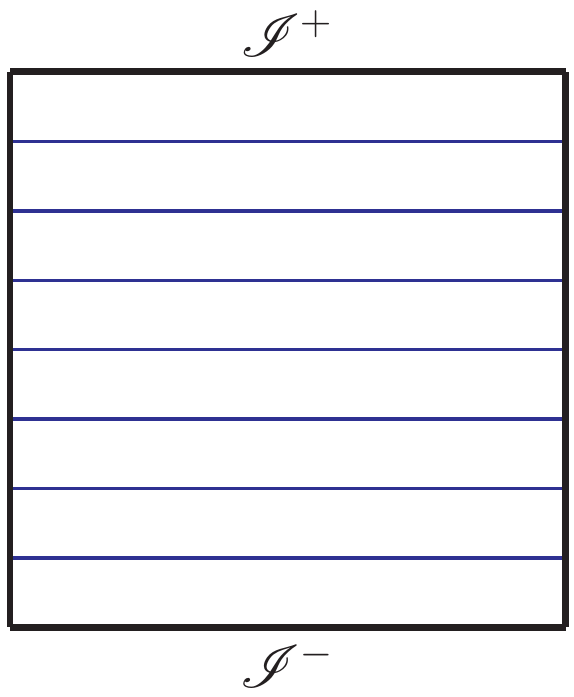}%
  \includegraphics[width=0.3\textwidth]{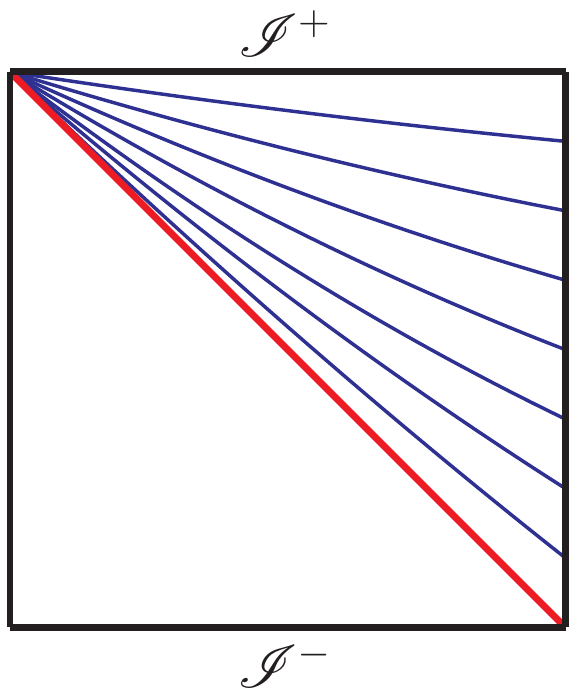}%
  \includegraphics[width=0.3\textwidth]{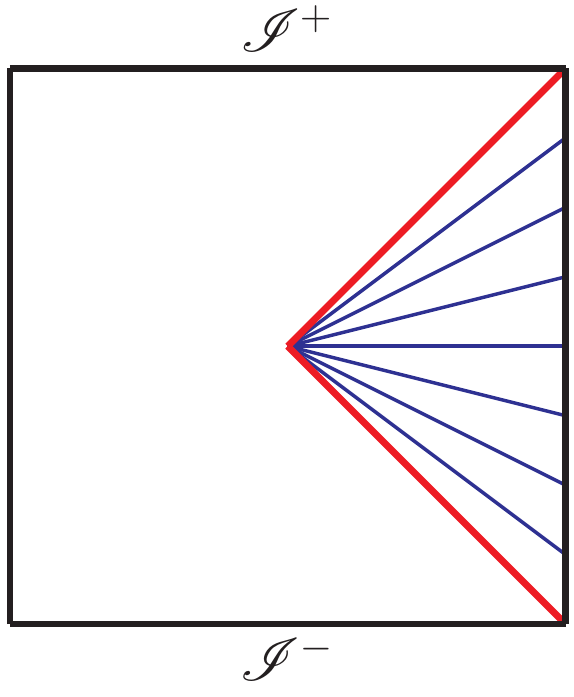}
  \caption{De Sitter coordinate charts depicted on the Carter-Penrose
    diagram.
    {\bf Left:} The global de Sitter chart with constant $t$ hypersurfaces
    depicted by blue lines. 
    {\bf Center:} The Poincar\'e chart with constant $\eta$ hypersurfaces 
    depicted by blue lines. This chart covers only the
    region above the cosmological horizon (solid red line).
    {\bf Right:} The static chart with constant $u$ hypersurfaces 
    depicted by blue lines. This chart covers only the region enclosed
    by past and future cosmological horizons (solid red lines).
    \label{fig:charts}
  }
\end{figure}

de Sitter space may be described by a number of coordinate charts,
the most common of which include the 
global, Poincar\'e (a.k.a.~cosmological), and static charts.
Although these charts cover different portions of the full de Sitter
manifold each chart defines a globally hyperbolic spacetime in
its own right.
Each chart is depicted on the de Sitter Carter-Penrose diagram
in Fig.~\ref{fig:charts}.
The global chart covers the entire manifold and is given by the 
line element
\eq{ \label{eq:gGlobal}
  \frac{ds^2}{\ell^2} = - dt^2 + (\cosh t)^2 d\Omega_{D-1}^2,
  \quad t \in \Reals .
}
Here $d\Omega_{D-1}^2$ is the line element on a unit $S^{D-1}$.
In this chart equal-time hypersurfaces are spheres of radius 
$\ell^2 (\cosh t)^2$. 
The Poincar\'e chart is given by
\eq{ \label{eq:gPoincare}
  \frac{ds^2}{\ell^2} = \frac{1}{\eta^2} \left[ - d\eta^2 
    + d\vec{x}\cdot d\vec{x} \right], 
  \quad \eta \in (-\infty, 0) ,
}
where $\vx$ denote vectors in $\Reals^{D-1}$.
This chart covers half of the full de Sitter manifold and
describes a cosmology with open spatial slices which expand exponentially
with increasing proper time. 
Finally, the static chart of de Sitter may be described by the line 
element
\eq{ \label{eq:gStatic}
  \frac{ds^2}{\ell^2} = 
  - \cos^2\theta du^2 + d\theta^2 + \sin^2\theta d\Omega_{D-2}^2,
  \quad u \in \Reals, \; \theta \in \left[ 0, \frac{\pi}{2} \right).
}
The static chart covers an even smaller region of de Sitter:
it is the largest region in which a de Sitter ``boost'' Killing vector
field $\d_u$ is timelike and future-directed.
Further discussion of de Sitter charts as well as other basic
features of de Sitter may be found in
\cite{Hawking:1973uf,Birrell:1982ix,Spradlin:2001pw}.

In order to conduct physics on de Sitter we need measures of 
distance and the causal structure.
As usual an invariant notion of distance between two points
$x_1,x_2 \in dS_D$ is provided by the signed, squared geodesic distance 
${\rm geod}(x_1,x_2)$, though it is more convenient to package
this information in the $SO(D,1)$-invariant 
``embedding distance'' \cite{Allen:1985ux}:
\eqn{ \label{eq:Z}
  Z(x_1,x_2) &:=& \frac{X_1 \cdot X_2}{\ell^2} 
  \nn \\
  &=& \left\{ 
    \begin{array}{lll}
    \cos\left[\ell^{-1}\sqrt{|{\rm geod}(x_1,x_2)|}\right] & \quad
    \text{spacelike separation} \\
    \cosh\left[\ell^{-1}\sqrt{|{\rm geod}(x_1,x_2)|}\right] & \quad
    \text{timelike, achronal separation} \\
    \end{array} 
    \right. .
}
The embedding distance satisfies
\begin{enumerate}[i)]
    \item $Z(x_1,x_2) \in [-1,1)$ for spacelike separation,
    \item $Z(x_1,x_2) = 1$ for null separation, 
    \item $Z(x_1,x_2) > 1$ for timelike separation, and
    \item $Z(x_1,x_2) < -1$ for achronal separation.
\end{enumerate}
The embedding distance is shown for various configurations
in Fig.~\ref{fig:Z}. In order to describe the causal structure
it is useful to introduce the function $s(x_1,x_2)$:
\eq{ \label{eq:s}
  s(x_1,x_2) := \left\{ 
    \begin{array}{lll}
      +1 & \quad & \text{if}\; x_1 \in J^+(x_2) \\
      -1 & \quad & \text{if}\; x_1 \in J^-(x_2) \\
      0  & \quad & \text{else} \\
      \end{array} \right. ,
}
where  $J^{+(-)}(x)$ denotes the causal future (past) of $x$.
Therefore the quantity $Z(x_1,x_2) - i \epsilon s(x_1,x_2)$,
which is invariant under $SO_0(D,1)$ rather than the
full isometry group $SO(D,1)$,
encodes both the geodesic distance between points as well as their 
causal relationship.
The following expressions give the embedding distance in the coordinate
charts described above:
\eqn{
  Z(x_1,x_2) &=& \left\{\begin{array}{ll}
      - \sinh t_1  \sinh t_2 + \cosh t_1 \cosh t_2 \cos \Omega_{12}  
      & ({\rm global}) \\
      1 - \frac{|\vec{x}_1 - \vec{x}_2|^2 - (\eta_1 - \eta_2)^2}{2\eta_1\eta_2}
      & (\text{Poincar\'e}) \\
      \cos\theta_1 \cos\theta_2 \cosh(u_1 - u_2) 
      + \sin\theta_1\sin\theta_2
      \cos \omega_{12} 
      & ({\rm static})
      \end{array}\right. . \nn \\
}
In these expressions $\Omega_{12}$ is the angular separation on $S^{D-1}$
of the global chart and $\omega_{12}$ is the angular separation 
on $S^{D-2}$ of the static chart respectively.

\begin{figure}[h!]
  \centering
  \includegraphics[width=0.5\textwidth]{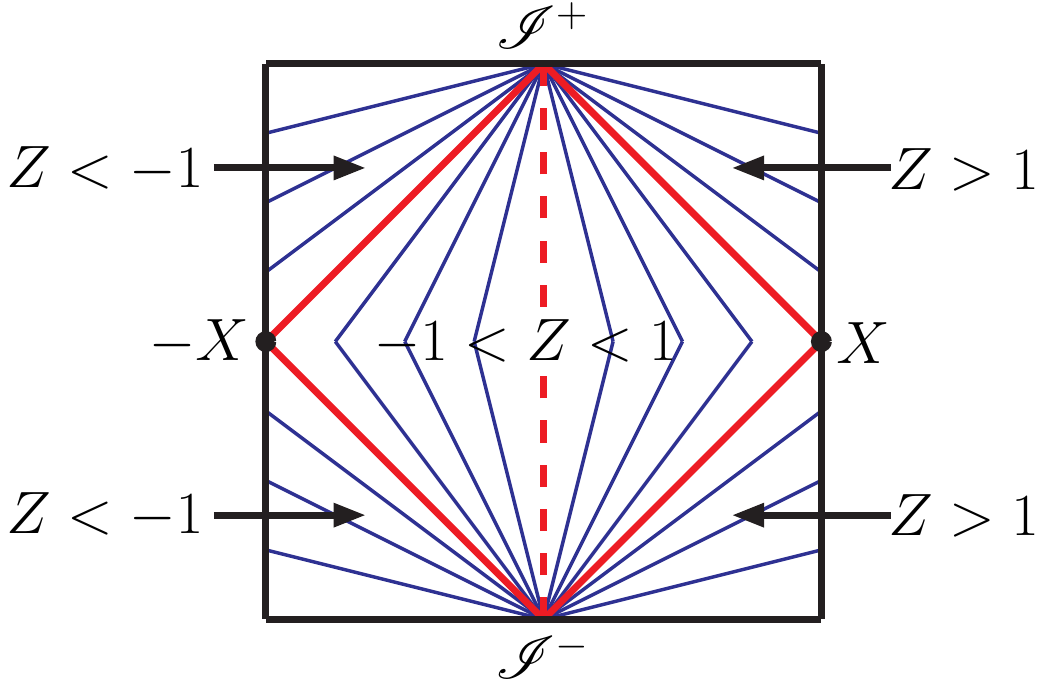}%
  \caption{The embedding distance $Z(x,y)$ on de Sitter.
    The point $x$ (equivalently $X$) and its
    antipodal point $-X$ are shown by dots. Thin blue lines
    depict constant $Z(x,y)$ contours, thick red lines denote
    the contours $Z(x,y)=\pm 1$, and the dashed red line denotes
    $Z(x,y) = 0$.
    \label{fig:Z}
  }
\end{figure}

\subsection{Linearized gravity}
\label{sec:gravity}

To obtain the theory of linear metric perturbations on de Sitter
space we let
\eq{
  g_{\mu\nu}(x) \to g_{\mu\nu}(x) + \sqrt{8\pi G} h_{\mu\nu}(x) ,
}
where $g_{\mu\nu}(x)$ is a de Sitter metric and $h_{\mu\nu}(x)$ a symmetric
perturbation.
All tensor indices are raised/lowered with the background metric 
$g_{\mu\nu}(x)$. The linearized equations of motion are\footnote{
  The authors of \cite{Higuchi:2012vy,Fewster:2012aa} refer to
  $L^{(1)}_{\mu\nu}(h)$ defined in (\ref{eq:L1}) as the linearized 
  Einstein tensor.}
\eqn{ \label{eq:L1}
  L^{(1)}_{\mu\nu}(h) 
  &:=& G^{(1)}_{\mu\nu}(h) + \Lambda h_{\mu\nu}(x)
  \nn \\
  &=& \half \Box h_{\mu\nu}(x)
  +\half \nabla_\mu\nabla_\nu h(x) 
  - \nabla_{(\mu}\nabla^\l h_{\nu)\l}(x)
  - \frac{1}{\ell^2} h_{\mu\nu}(x)
  \nn \\ & &
  + \half \left(\nabla^\a\nabla^\b h_{\a\b}(x)
    - \Box h(x) + \frac{3-D}{\ell^2} h(x) \right) g_{\mu\nu}(x) 
  \nn \\
  &=& 0 ,
}
where $h(x) = h^\nu_{\phantom{\nu}\nu}(x)$.\footnote{When there
  is no risk of confusion we will omit tensor indices on arguments.
  This does not mean the argument is simply the trace of the tensor.
}
The operator $L^{(1)}_{\mu\nu}(h)$ is symmetric, transverse 
$\nabla^\mu L^{(1)}_{\mu\nu}(h) = 0$, and linear in the usual sense:
$L^{(1)}_{\mu\nu}(\alpha h_{\mu\nu} + \beta \gamma_{\mu\nu})
= \alpha L^{(1)}_{\mu\nu}(h) + \beta L^{(1)}_{\mu\nu}(\gamma)$
for $\a,\beta \in \mathbb{C}$. 
It is also self-adjoint in the sense that for tensors
$f_{\a\b}(x)$, $q_{\mu\nu}(x)$ such that $f_{\a\b}(x) q_{\mu\nu}(x)$ has
compact support then
\eq{ \label{eq:Ladjoint}
  \int d^D \sqrt{-g} \, q^{\mu\nu}(x) L^{(1)}_{\mu\nu}(f) 
  = \int d^D \sqrt{-g} \, f^{\mu\nu}(x) L^{(1)}_{\mu\nu}(q) .
}

The equations of motion are left invariant under the field shift
\eq{ \label{eq:gaugeXF}
  h_{\mu\nu}(x) \to h_{\mu\nu}(x) + \Lie_\xi g_{\mu\nu}(x)
  =  h_{\mu\nu}(x) + 2\nabla_{(\mu} \xi_{\nu)}(x) ,
}
for any 1-form $\xi_\mu(x)$. We refer to (\ref{eq:gaugeXF}) as the 
gauge symmetry of linearized gravity.
As for curvature quantities, it follows that 
the linearized Einstein tensor $G_{\mu\nu}^{(1)}(h)$ is not gauge-invariant.
The linearized Weyl tensor $C^{(1)\a}_{\ph{(1)\a}\b\gamma\delta}(h)$ is
gauge invariant
\eq{ \label{eq:Weylgrad}
  C^{(1)\a}_{\ph{(1)\a}\b\gamma\delta}(\Lie_\xi g) = 0 ,
}
and is also invariant under linearized Weyl transformations
\eq{ \label{eq:Weylmet}
   C^{(1)\a}_{\ph{(1)\a}\b\gamma\delta}(\omega g_{\mu\nu}) = 0 ,
}
for arbitrary $\omega(x)$. Because the Weyl tensor of the de Sitter
background vanishes these equations are valid for any index configuration.
An explicit formula for $C^{(1)\a\b}_{\ph{(1)\a\b}\gamma\delta}(h)$ 
is
\eqn{ \label{eq:Weyl}
  C^{(1)\a\b}_{\ph{(1)\a\b}\gamma\delta}(h)
  &:=& \sqrt{8\pi G}\, \Upsilon^{[\a\b]\,[\rho\s]}_{[\gamma\delta]\,[\mu\nu]} 
  \left( \nabla_\rho\nabla^\mu 
      + \ell^{-2} \delta^\mu_\rho \right) h^\nu_{\ph{\nu}\s} ,
  \\
  \label{eq:Upsilon}
  \Upsilon^{\a\b\,\rho\s}_{\gamma\delta\,\mu\nu} &:=&
  -2 \left( 
    \delta^\a_\mu \delta^\b_\nu \delta^\rho_\gamma \delta^\s_\delta
    + \frac{4}{(D-2)}
    \delta^\a_\mu \delta^\b_\gamma \delta^\rho_\delta \delta^\s_\nu
    + \frac{2}{(D-2)(D-1)}
    \delta^\a_\gamma \delta^\b_\delta \delta^\rho_\mu \delta^\s_\nu
    \right) . \nn \\
}
Many useful lemmas, theorems, are formulae for linearized gravity in 
vacuum cosmological spacetimes are presented in \cite{Fewster:2012aa}.

\subsection{Quantization}
\label{sec:quantum}

In this section we outline our quantization procedure
which is an example of the algebraic approach to QFT in curved spacetime 
\cite{Wald:1995yp}.\footnote{A thorough treatment of the 
algebraic approach to QFT in Minkowski space is presented in
\cite{Haag:1992aa}.}
Although there are several advantages to this approach
(see e.g. the discussions in \cite{Wald:2006ty,Kay:2006jn,Hollands:2008vx}),
the immediate advantage for our purposes is that it
allows us to quantize linearized gravity in covariant gauges
which in turn allow us to preserve manifest de Sitter symmetry throughout
our computations.
These gauges are not a complete gauge fixing and are analogous to
Lorentz gauge in vector gauge theories.
The algebraic quantization of linearized gravity on cosmological
backgrounds has been treated previously by
Fewster and Hunt \cite{Fewster:2012aa}.

A second advantage of the algebraic approach is that it allows us
to discuss in a unified manner the quantization of gravity in each 
of the de Sitter backgrounds described in \S\ref{sec:dS}.
Other tactics such as path integral or ``mode
quantization'' approaches, if performed in any detail, would require us 
to work separately on each background as each admits different 
sets of solutions to the classical equations of motion.
Various technical issues then arise in comparing the resulting Hilbert
spaces. In any case, we will refer the background de Sitter 
spacetime simply as ``$dS_D$'' which implicitly includes a choice of
chart and will make chart-specific comments as needed.
We also adopt the notation
\eq{
  \int_x F(x) := \int d^Dx \sqrt{-g(x)} F(x) ,
}
where similarly the volume element is that of the background of interest.
The integrands of such expressions will always have compact support. 

\subsubsection{Observables}

Local observables of the quantum theory are the analogues of local, 
gauge-invariant quantities of the classical theory.
Observables constructed directly from the
metric perturbation are of the form \cite{Fewster:2012aa}
\eq{ \label{eq:hObs}
  h(f) := \int_x
  f^{\mu\nu}(x) h_{\mu\nu}(x), \quad f \in \testT ,
}
where $f^{\mu\nu}(x)$ is a ``test function'' that is compactly supported, 
symmetric, and transverse, i.e., it belongs to the class\footnote{Test
function tensor indices may be freely raised and lowered with the
de Sitter background metric, so the choice to define test functions 
as $T^2_0$ tensors rather than $T^0_2$ is arbitrary. We will typically 
adopt covariant/contravariant indices as is convenient for notational 
clarity.}
\eq{ \label{eq:testT}
  \testT := \left\{ f^{\mu\nu} \in C^\infty_0(T^2_0(dS_D)) \; | \;    
  f^{\mu\nu} = f^{\nu\mu}, \; \nabla_\mu f^{\mu\nu} = 0 \right\} .
}
In principal compact support can include support
on an entire $S^{D-1}$ Cauchy surface in global dS, so long as
the support in the timelike direction is compact.
The compactness and transversality of $f^{\mu\nu}(x)$ guarantee that 
$h(f)$ is gauge-invariant. In addition these observables satisfy
\begin{enumerate}[O1.]
  \item\label{O:linear} Linearity:
    $\quad
      h(\alpha f^{\mu\nu} + \beta p^{\mu\nu})
      = \alpha h(f) + \beta h(p) , \quad \a,\b \in \mathbb{C},
    $
  \item\label{O:hermiticity} Hermiticity:
    $\quad
      h^\dagger(f) = h(f^*) ,
    $
  \item\label{O:EOM} Equations of motion:
    $\quad
      h(L^{(1)}(f)) = 0 ,
    $
  \item\label{O:locality} Locality (canonical commutation relations):
    $\quad
      \left[ h(f), h(p) \right] = i \Delta(f, p) ,
    $
\end{enumerate}
where $\Delta(f,p)$ is the smeared commutator function, i.e. the 
unique advanced-minus-retarded solution to the equations of motion
\cite{Friedlander:1975aa}.

One may also consider observables constructed from the linearized
Weyl tensor:
\eq{ \label{eq:CObs}
  C^{(1)}(v) := \int_x v^{\mu\nu\rho\sigma}(x) 
  C^{(1)}_{\mu\nu\rho\sigma}(h),
  \quad v \in C^\infty_0(T^4_0(dS_D)) .
}
Smeared ``Wick powers'' of the linearized Weyl tensor
such as $C^{(1)}_{\a\b\gamma\delta} C^{(1)\mu\nu\rho\s}(x)$, 
$C^{(1)}_{\a\b\gamma\delta} C^{(1)\gamma\delta\mu\nu}(x)$, etc., are 
also local observables provided an ordering prescription for defining
composite operators (such as that of \cite{Hollands:2001fk} or 
\cite{Brunetti:1995rf}).

Consider the set of observables of the forms (\ref{eq:hObs}), (\ref{eq:CObs}),
as well as their associated gauge-invariant Wick powers,
all constructed from test functions whose support is contained in the
compact region $\cO \subset dS_D$.
Along with the identity element, the polynomial algebra generated by 
finite sums of finite products of members of this set form a 
unital *-algebra $\cA(\cO)$; the union of the algebras of all 
such open sets on $dS_D$ to define the total algebra of local observables
\eq{
  \cA(dS_D) = \cup_\cO \,\cA(\cO) .
}

\subsubsection{States}

A quantum state $\Psi$ is a linear functional on the algebra of 
observables $\Psi: \cA(dS_D) \to \Complex$ which additionally satisfies
\begin{enumerate}[S1.]
  \item \label{S:norm}
    Normalization: $\quad\C{1}_\Psi = 1$, 
  \item \label{S:pos}
    Positivity: (a.k.a. ``unitarity''):
    $\quad
      0 \le \C{ A^\dagger A }_\Psi < \infty ,
      \quad \forall \; A \in \cA(dS_D) ,
    $
  \item $\Psi$ must be of the Hadamard type.
\end{enumerate}
The first two requirements are familiar and are necessary
for $\Psi$ to admit a quantum mechanical interpretation as 
providing conditional probabilities.
The third criteria is a choice; it is a regularity 
condition which assures that the states we consider
i) reproduce familiar Minkowski space physics at distances much less
than the de Sitter radius, and
ii) allow local and covariant techniques for defining the 
Wick powers described below (\ref{eq:CObs}).
These demands drive us to consider the class of states
which are ``globally Hadamard'' \cite{Wald:1995yp}.\footnote{A 
  simple example of non-Hadamard states are the
  Mottola-Allen or $\alpha$-vacua \cite{Mottola:1984ar,Allen:1985ux}
  (excepting the Euclidean state).
  The correlation functions of these states 
  posses singularities at spacelike separations  and differ in the 
  character of their lightcone singularity from that of the standard 
  Minkowski vacuum at arbitrarily short distances \cite{Bousso:2001mw}. 
  For these states normal ordering does not yield well-defined
  composite objects \cite{Brunetti:2005pr}.
  Of course, the ``states of interest'' depend on the context, and
  interesting things have been done with these states
  \cite{Lagogiannis:2011st}.
}

For scalar field theory a quasi-free Hadamard state is one for which
the 2-pt function is singular only at null separations, and moreover
this singularity is pure ``positive frequency.''
More precisely, the nature of a distribution's singularity may be
described by its wave front set (WF) \cite{Hormander:1990aa}.
A scalar 2-pt function $\Delta(x_1,x_2)$ is Hadamard if its 
wave front set is given by \cite{Radzikowski:1996aa}
\eq{ \label{eq:HadamardWF}
  \WF( \Delta )
  = \{ (x_1, k_1 ; x_2, -k_2) \in (T^*dS_D \setminus \{0\})^2 \; | 
  \; (x_1,k_1)\sim (x_2,k_2),
  \; k_1 \in V_1^+ \} ,
}
where $(x_1,k_1)\sim (x_2,k_2)$ denotes that $x_1$ and $x_2$ may
be joined by a null geodesic $k_1$ and $k_2$ are cotangent and
coparallel to that geodesic.
In this paper we will deal with tensor-valued fields.
We will refer to a graviton state as Hadamard if, in addition to
being quasi-free, the state's 2-pt function wave front set is 
given by (\ref{eq:HadamardWF}),
where the wave front set of tensor distribution is simply the
union of the the wave front sets of its components in a
local trivialization \cite{Hormander:1990aa}.
This definition is sufficient to guarantee the desired properties
i) and ii) above and is the natural generalization of the Hadamard 
condition for vector fields
\cite{Sahlmann:2001aa,Dappiaggi:2011cj}. 
The Hadamard condition for tensor distributions may also be formulated
in terms of ``Hadamard's fundamental solutions'' 
\cite{Allen:1988aa,DeWitt:1960aa}.

\subsubsection{Gauge symmetry} 
\label{sec:gauge}

Let us now address the gauge symmetry of the quantum theory. 
We would like the qauntum theory to admit
the same gauge symmetry as the classical equations of motion;
in particular, we consider any field redefinition of the form
$h_{\mu\nu}(x) \to h_{\mu\nu}(x) + 2\nabla_{(\mu}\xi_{\nu)}(x)$
for \emph{any} smooth 1-form field $\xi_{\nu}(x)$ to be a valid
gauge transformation. The observables constructed above
are indeed invariant under any such redefinition.
In the cosmology literature it is common to focus on quantization
on the Poincar\'e chart, and to restrict 
the gauge freedom to $\nabla_{(\mu}\xi_{\nu)}(x)$ which vanish 
sufficiently rapidly near the spatial conformal boundaries of
this chart. We do not impose any such restriction.

In order to construct graviton 2-pt functions in a manner preserving
as much symmetry as possible we adopt
covariant gauge conditions. These gauge conditions do not fully
fix the gauge redundancy and so may be thought of as ``partial''
gauge conditions which are sufficient to allow us to invert the 
equations of motion.\footnote{This level of gauge fixing
is insufficient to render a path integral formulation well-defined.
Presumably standard procedures such as the Stuckelburg formulation
can be used to formulate a path-integral in the covariant gauges
we consider, but we do not pursue this here.}
We will utilize two types of gauges. The first is the class of 
generalized de Donder gauges \cite{Mora:2012zi} which satisfy
the gauge condition
\eq{ \label{eq:dD}
  \nabla^\nu h_{\mu\nu}(x) -\frac{\beta}{2} \nabla_\mu h(x) = 0 ,
  \quad \beta \in \Reals.
}
The choice $\beta = 1$ it typically referred to as de Donder or
harmonic gauge while $\beta = 0$ is transverse gauge.
At least for generic values of $\beta$ we expect
to be able to impose this gauge condition non-linearly.
The de Donder gauges (\ref{eq:dD}) do not completely fix the gauge
freedom as it is possible to construct vector fields $\xi^\mu(x)$ such
that $\nabla_{(\mu} \xi_{\nu)}(x)$ satisfy (\ref{eq:dD}).

In linearized gravity one may use the linearized equations of motion
to impose further gauge conditions; in particular, one may impose
that solutions to the equations of motion be in transverse traceless (TT)
gauge:
\eq{ \label{eq:TTgauge}
  \nabla^\nu h_{\mu\nu}(x) = 0, \quad h(x) = 0 .
}
On dS there exist vector fields $\xi^\mu(x)$ such
that $\nabla_{(\mu} \xi_{\nu)}(x)$ satisfy (\ref{eq:TTgauge})
(see, e.g., \cite{Higuchi:1991tn} or Appendix F of \cite{Park:2008ki}), so
there still exists residual gauge symmetry in TT gauge.\footnote{The gauge
may be fixed completely by imposing, e.g., transverse traceless synchronous
gauge where in addition to (\ref{eq:TTgauge}) one imposes
$h_{t\mu}(x) = 0$ for some time coordinate $t$, but this introduces 
a preferred timelike direction which is undesirable for our investigation.}
For most of our analysis below we adopt TT gauge as we work exclusively
at in the linearized theory and this gauge makes our analysis rather
simple. In \S\ref{sec:deDonder} we compute the 2-pt function of the
state $\Omega$ in the de Donder gauges as well so that it may be
utilized in non-linear perturbation theory.

A very useful fact is that in any gauge, when acting on solutions 
to the linearized  equations of motion, the set of test functions 
$\testT$ defined in (\ref{eq:testT}) may be further restricted to 
the class of TT test functions \cite{Higuchi:2012vy}
\eq{ \label{eq:testTT}
  \testTT := \left\{ f^{\mu\nu} \in C^\infty_0(T_0^2(dS_D)) \; | \;    
    f^{\mu\nu} = f^{\nu\mu}, \; \nabla_\mu f^{\mu\nu} = 0, \;
    g_{\mu\nu} f^{\mu\nu} = 0 \right\} .
}
That is, for every $f \in \testT$ there exists a $p \in \testTT$
such that $\C{h(f)\dots}_\Psi = \C{h(p)\dots}_\Psi$.

\section{The state $\Omega$}
\label{sec:2pt}

In this section we construct the maximally-symmetric state $\Omega$ 
by computing the 2-pt function $\CO{h_{\mu\nu}(x)h^{\bmu\bnu}(\bx)}$
in TT gauge. 
Our derivation of $\CO{h_{\mu\nu}(x)h^{\bmu\bnu}(\bx)}$ is straight-forward 
but it utilizes a great amount of simple technology which we spend the
next three subsections describing. We finally compute 
$\CO{h_{\mu\nu}(x)h^{\bmu\bnu}(\bx)}$ in \S\ref{sec:TT}, as well as
verify the Hadamard and positivity properties of this 2-pt function.
We finish this section by contrasting our result with earlier works
in \S\ref{sec:compare}. Henceforth we set the
de Sitter radius $\ell = 1$.

\subsection{Transverse traceless projection operator}
\label{sec:P}

A convenient way to impose the transverse and traceless conditions
on the metric perturbation is via
a transverse traceless projection operator \cite{Miao:2011fc}.
It is natural to construct this operator from the linearized
Weyl tensor; from the symmetries of the Weyl tensor it follows that the
operation $\nabla^\gamma\nabla^\delta C^{(1)}_{\gamma\mu\delta\nu}(h)$
constructs from any symmetric tensor $h_{\a\b}(x)$ a symmetric, rank-2 TT 
tensor. Therefore we define the TT projection operator
$\P_{\mu\nu}^{\phantom{\mu\nu}\alpha\beta}$ via
\eq{
  \P_{\mu\nu}^{\phantom{\mu\nu}\alpha\beta} h_{\a\b}
  := \nabla^\gamma\nabla^\delta C^{(1)}_{\gamma\mu\delta\nu}(h) .
}
This operator is symmetric on each pair of indices
$\P_{\mu\nu}^{\phantom{\mu\nu}\alpha\beta} = 
\P_{\mu\nu}^{\phantom{\mu\nu}(\alpha\beta)} = 
\P_{(\mu\nu)}^{\phantom{(\mu\nu)}\alpha\beta}$ and is transverse and 
traceless on indices $\mu$,$\nu$:
\eq{
  \nabla^\nu \P_{\mu\nu}^{\phantom{\mu\nu}\alpha\beta} f_{\alpha\beta} = 0, 
  \quad 
  g^{\mu\nu} \P_{\mu\nu}^{\phantom{\mu\nu}\alpha\beta} f_{\alpha\beta} = 0 .
}
Additional properties of $\P_{\mu\nu}^{\phantom{\mu\nu}\alpha\beta}$ we will
need are \cite{Miao:2011fc}: 
\begin{enumerate}[i)]
\item 
  the d'Alembertian 
  commutes with $\P_{\mu\nu}^{\phantom{\mu\nu}\alpha\beta}$, 
  i.e. $\Box \P_{\mu\nu}^{\phantom{\mu\nu}\a\b} f_{\a\b}(x)
  = \P_{\mu\nu}^{\phantom{\mu\nu}\a\b} \Box f_{\a\b}(x)$, 
\item 
  the action of $\P_{\mu\nu}^{\phantom{\mu\nu}\a\b}$ on TT tensors 
  is
  \eq{ \label{eq:PonTT}
    \P_{\mu\nu}^{\phantom{\mu\nu}\a\b} w_{\a\b}(x)
    =  - \half \frac{(D-3)}{(D-2)} (\Box - 2)(\Box-D) w_{\mu\nu}(x) , 
  }
\item
  if the product $p^{\mu\nu}(x) f_{\a\b}(x)$ is compactly supported
  then $\P_{\mu\nu}^{\phantom{\mu\nu}\a\b}$ is self-adjoint in the sense
  that
  \eq{ \label{eq:Padjoint}
    \int_x p^{\mu\nu}(x) \P_{\mu\nu}^{\phantom{\mu\nu}\a\b} f_{\a\b}(x)
    = \int_x f^{\mu\nu}(x) \P_{\mu\nu}^{\phantom{\mu\nu}\a\b} p_{\a\b}(x).
  }
\end{enumerate}

Obviously, the projection operator annihilates any tensor $q_{\mu\nu}(x)$ 
for which $\nabla^\mu\nabla^\s C^{(1)}_{\mu\nu\sigma\rho}(q) = 0$.
This includes total derivative and metric terms
\eq{ \label{eq:PonGradMet}
  \P_{\mu\nu}^{\phantom{\mu\nu}\alpha\beta} \nabla_{\alpha} f_\beta(x) = 0 ,
  \quad
  \P_{\mu\nu}^{\phantom{\mu\nu}\alpha\beta} (f(x) g_{\alpha\beta}) = 0 ,
}
as for these terms $C^{(1)}_{\a\b\gamma\delta}(q) = 0$ (recall 
(\ref{eq:Weylgrad}) and (\ref{eq:Weylmet})). 
In addition we determine from (\ref{eq:PonTT}) that
the projection operator annihilates TT solutions to
the de Sitter Fierz-Pauli equation \cite{Fierz:1939aa} 
\eq{
  (\Box - M^2-2) h_{\mu\nu}(x) = 0, 
}
(the factor of $2$ arising from the cosmological constant term
in the action)
with mass values of $M^2 = 0$ and $M^2 = (D-2)$. 
Solutions for $M^2 = 0$ are solutions to linearized
Einstein equations (\ref{eq:L1}) and satisfy 
$\nabla^\a C^{(1)}_{\a\b\gamma\delta}(q) = 0$; solutions for $M^2 = (D-2)$
 are sometimes called ``partially massless,'' correspond to the 
Higuchi lower bound for unitary massive spin-2 fields \cite{Higuchi:1986py},
and satisfy $\nabla^\a\nabla^\gamma C^{(1)}_{\a\b\gamma\delta}(q) = 0$.

\subsection{Maximally symmetric bi-tensors}
\label{sec:MSBTs}

It is convenient to work with maximally symmetric bi-tensors (MSBTs) 
as these manifestly preserve 
maximal symmetry (classic references on these objects include 
\cite{Allen:1985wd,Allen:1986qj,Allen:1986tt,Allen:1994yb}). The
tensor structures of MSBTs at $x$, $\bx$ are constructed by taking 
covariant derivatives of $Z := Z(x,\bx)$ in the tangent spaces of 
$x$ and $\bx$ respectively. 
For rank-2 symmetric MSBTs there are five allowed index structures:
\eqn{
  \label{eq:t1}
  t_{\mu\nu}^{(1)\; \bmu\bnu} &:=&
  g_{\mu\nu} g^{\bmu\bnu} ,
  \\
  \label{eq:t2}
  t_{\mu\nu}^{(2)\; \bmu\bnu} &:=&
  (\nabla_{(\mu}\nabla^{(\bmu}Z) (\nabla_{\nu)}\nabla^{\bnu)}Z) ,
  \\
  \label{eq:t3}
  t_{\mu\nu}^{(3)\; \bmu\bnu} &:=&
  (\nabla_{(\mu}Z)(\nabla^{(\bmu}Z) (\nabla_{\nu)}\nabla^{\bnu)}Z) ,
  \\
  \label{eq:t4}
  t_{\mu\nu}^{(4)\; \bmu\bnu} &:=&
  (\nabla_\mu Z)(\nabla_\nu Z)(\nabla^\bmu Z)(\nabla^\bnu Z) ,
  \\
  \label{eq:t5}
  t_{\mu\nu}^{(5)\; \bmu\bnu} &:=&
  \left[ g_{\mu\nu} (\nabla^\bmu Z)(\nabla^\bnu Z) + 
  (\nabla_\mu Z)(\nabla_\nu Z) g^{\bmu\bnu} \right] .
}
Any rank-2 symmetric MSBT may be written in terms of these tensors
with five scalar coefficient functions of $Z$, i.e.,
\eq{ \label{eq:dumb}
  M_{\mu\nu}^{\phantom{\mu\nu}\bmu\bnu}(Z)
  = \sum_{i=1}^5 a_i(Z) t_{\mu\nu}^{(i)\; \bmu\bnu} .
}
However, often this is not the most convenient way of organizing
the five scalar functions which determine 
$M_{\mu\nu}^{\phantom{\mu\nu}\bmu\bnu}(Z)$; following \cite{DHoker:1999aa}, 
we note that the most general $M_{\mu\nu}^{\phantom{\mu\nu}\bmu\bnu}(Z)$
may also be written in the forms
\eqn{ \label{eq:sly}
  M_{\mu\nu}^{\phantom{\mu\nu}\bmu\bnu}(Z)
  = b_j(Z) t_{\mu\nu}^{(j)\; \bmu\bnu} 
  + b_1(Z) t_{\mu\nu}^{(1)\; \bmu\bnu} 
  + \sum_{i=1}^3 G^{(i)\; \bmu\bnu}_{\mu\nu}(c_i(Z)) ,
  \quad j = 2 \; {\rm or} \; 4 ,\quad\quad
}
where the $G^{(i)\; \bmu\bnu}_{\mu\nu}(c_i(Z))$ are MSBTs which
are total derivatives at $x$ and/or $\bx$:
\eqn{
  G^{(1) \; \bmu\bnu}_{\mu\nu}(c_1(Z)) 
  &:=& \left[ g_{\mu\nu}\nabla^\bmu\nabla^\bnu +
    g^{\bmu\bnu}\nabla_\mu\nabla_\nu \right] c_1(Z) ,
  \label{eq:G1}
  \\
  G^{(2) \; \bmu\bnu}_{\mu\nu}(c_2(Z)) 
  &:=& \nabla_{(\mu} 
  \left[ c_2(Z) (\nabla^{(\bmu} Z) (\nabla_{\nu)} \nabla^{\bnu)} Z) \right]
  + (x \leftrightarrow \bx) ,
  \label{eq:G2}
  \\
  G^{(3) \; \bmu\bnu}_{\mu\nu}(c_3(Z)) 
  &:=& \nabla_{(\mu} 
  \left[ c_3(Z) (\nabla_{\nu)} Z) (\nabla^\bmu Z)(\nabla^\bnu Z) \right]
  + (x \leftrightarrow \bx) .
  \label{eq:G3}
}
By expanding the expressions (\ref{eq:G1})-(\ref{eq:G3}) one may
readily verify that $b_j(Z)$, $b_1(Z)$, and the three $c_i(Z)$ uniquely
determine the five $a_i(Z)$ in (\ref{eq:dumb}).

Tensor indices belonging the same tangent space are easily 
contracted and may be simplified by noting that 
$\nabla_\mu X^A$, where $X^A$ with $A=0,\dots,D$ are Cartesian coordinates
in the embedding space, are conformal Killing vectors on $dS_D$:
\eq{
  \nabla_\mu\nabla_\nu X^A = - X^A g_{\mu\nu}, \quad \Rightarrow \quad
  \nabla_\mu \nabla_\nu Z = - Z g_{\mu \nu} ,
}
and as a result $Z$ is a maximally symmetric bi-scalar 
eigenfunction of the d'Alembertian:
\eq{ \label{eq:BoxZ}
  \Box Z = \bBox Z = -D Z .
}
Other useful contractions that follow from (\ref{eq:BoxZ}) include:
\eqn{ \label{eq:Zids}
  (\nabla^\mu Z) (\nabla_\mu Z) &=& (1-Z^2) , \nn \\
  (\nabla^\mu Z) (\nabla_\mu\nabla_\bnu Z) &=& -Z(\nabla_\bnu Z) , \nn \\
  (\nabla^\mu Z) (\nabla^\bnu Z) (\nabla_\mu\nabla_\bnu Z) , 
  &=& -Z(1-Z^2) .
}

MSBTs may also be parametrized in terms of the parallel propagator 
$g_{\mu\bnu}(Z)$ and the unit normal vector $n_\mu(Z)$ which is
tangent to the shortest geodesic between $x$ and $\bx$\cite{Allen:1985wd}. 
These objects are defined in terms of derivatives of $Z$ as follows:
\eqn{ \label{eq:pp}
  g_{\mu\bnu}(Z) 
  &:=& (1+Z) \nabla_\mu\nabla_\bnu \ln(1+Z)
  = \nabla_\mu\nabla_\bnu Z 
  - \frac{1}{(1+Z)} (\nabla_\mu Z)(\nabla_\bnu Z) .
 \\
 \label{eq:n}
 n_\mu(Z) &:=& \frac{-1}{(1-Z^2)^{1/2}} \nabla_\mu Z .
}
An advantage of these variables is that they are bounded
functions of $Z$; a disadvantage is that their covariant 
derivatives are cumbersome.
Using the parallel propagator one may translate the tensor structure
of a tensor in the tangent space of $\bx$ to that of $x$, but the 
resulting object is still a function of $\bx$, i.e.,
$g_{\mu}^{\phantom{\mu}\bnu}V_\bnu(\bx)$ defines a 1-form $V_\mu(\bx)$
which is not in general equivalent to $V_\mu(x)$.

Often we will only need to know an MSBT modulo terms proportional
to the metric or terms which are total derivatives at $x$ or $\bx$.
We refer to such terms as `$\met$' and `$\grad$' terms respectively.
For instance, we may write the most general MSBT (\ref{eq:sly}) as
\eq{ \label{eq:sly2}
  M_{\mu\nu}^{\phantom{\mu\nu}\bmu\bnu}(Z)
  = b_2(Z) \T2{\mu}{\nu}{\bmu}{\bnu} + \met + \grad .
}
Later it will be useful to note that
\eq{ \label{eq:T2gauge}
  \T2{\mu}{\nu}{\bmu}{\bnu} = \met + \grad, \quad
  Z \T2{\mu}{\nu}{\bmu}{\bnu} = \met + \grad ,
}
and that these are the only tensors with index structure 
$\T2{\mu}{\nu}{\bmu}{\bnu}$ which are ``pure $\met + \grad$.''

\subsection{K\"allen-Lehmann representations}
\label{sec:scalar}

We will write the graviton 2-pt function in terms of relatively
simple scalar functions which are closely related to the 
Green's functions of the scalar Klein-Gordon equation.
It is convenient to employ a kind of K\"allen-Lehmann representation 
for these scalar functions (utilized previously by
\cite{Marolf:2010zp,Hollands:2010pr,Hollands:2011we,Marolf:2012kh}). 
We review this K\"allen-Lehmann representation now.

The maximally-symmetric Green's functions to the Klein-Gordon equation
were obtained in ancient times by Motolla \cite{Mottola:1984ar} and Allen 
\cite{Allen:1985ux}. Consider the function
\eq{ \label{eq:Delta}
  \Delta_\sigma(Z) := 
  \frac{1}{(4\pi)^{D/2}}
  \frac{\Gamma(-\sigma)\Gamma(\sigma+D-1)}
  {\Gamma\left(\frac{D}{2}\right)}
  \,{}_2F_1\hspace{-4pt}\left[-\sigma, \sigma+D-1 ; \frac{D}{2} ; 
    \frac{1+Z}{2} \right] ,
}
where ${}_2F_1[a,b;c;z]$ is the Gauss hypergeometric 
function \cite{Bateman:1955}. Eq. (\ref{eq:Delta}) defines
an analytic function of $Z$ in the complex $Z$ plane cut along 
$Z \in [1,\infty)$; it is also a meromorphic function of $\sigma$ with
simple poles at $\s = 0,1,2,\dots$ and $\s = -(D-1), -D, -(D+1),\dots$.
We obtain Green's functions by adding various $i\epsilon$ prescriptions
for avoiding the cut in the $Z$ plane. Relating the parameter $\sigma$ to 
the mass via
\eq{
  M^2 = M^2(\s) := -\s(\s+D-1) ,
}
the distributions $\Delta_\s(Z - i\epsilon)$ and 
$\Delta_\s(Z + i\epsilon s)$ satisfy
\eq{
  (\Box - M^2) \Delta_\s(Z - i\epsilon) = i\frac{\delta^D(x,\bx)}{\sqrt{-g}},
  \quad 
  (\Box - M^2) \Delta_\s(Z - i\epsilon s) = 0 ,
}
where $s = s(x,\bx)$ as in (\ref{eq:s}). 
The normalization of $\Delta_\s(Z)$ has been 
chosen such that for $M^2(\s) > 0$ these objects correspond to the 
time- and Wightman-ordered 2-pt functions of a canonically 
normalized massive scalar field on $dS_D$.\footnote{When
$M^2(\s)$ is not positive $\Delta_\s(Z-i\epsilon s)$ does not 
define a 2-pt function for Klein-Gordon fields on $dS_D$. 
For these values of the mass $\Delta_\s(Z-i\epsilon s)$
still satisfies the Klein-Gordon equation and is a Hadamard distribution,
but it fails to satisfy the positivity condition.}
$\Delta_\s(Z)$ is a scalar Hadamard distribution.

Perhaps the simplest example of a K\"allen-Lehmann representation
is that for $\Delta_\s(Z)$ itself \cite{Marolf:2010zp}:
\eq{ \label{eq:DeltaContour}
  \Delta_\s(Z) = \int_{C_\s} \frac{d\omega}{2\pi i}
  \frac{(2\omega+D-1)}{(\omega-\s)(\omega+\s+D-1)} \Delta_\omega(Z) .
}
The integration contour $C_\s$ is traversed from $-i\infty$ to $+i\infty$
within the strip $\Re \s < \Re \omega < 0$ -- see Fig.~\ref{fig:scalarContours}.
By acting with $(\Box - M^2)$ we obtain the identity
\eqn{ \label{eq:dirac}
  i\frac{\delta^D(x,\bx)}{\sqrt{-g}} 
  &=& \int_{C_\s} \frac{d\omega}{2\pi i}
  \frac{(2\omega+D-1)}{(\omega-\s)(\omega+\s+D-1)} 
  (\Box-M^2(\s))\Delta_\omega(Z-i\epsilon)
  \nn \\
  &=&
  \int_{C_\s} \frac{d\omega}{2\pi i}
  \frac{(2\omega+D-1)}{(\omega-\s)(\omega+\s+D-1)} 
  \bigg[ -(\omega-\s)(\omega+\s+D-1)\D_\omega(Z-i\epsilon) 
  \nn \\ & & \ph{  \int_{C_\s} \frac{d\omega}{2\pi i}
  \frac{(2\omega+D-1)}{(\omega-\s)(\omega+\s+D-1)} 
  \bigg[}
    + i\frac{\delta^D(x,\bx)}{\sqrt{-g}} \bigg]
  \nn \\
  &=&
  - \int_C \frac{d\omega}{2\pi i} 
  (2\omega+D-1) \Delta_\omega(Z-i\epsilon) .
}
In the second equality the second term on the right-hand side does
not contribute because for this term the integration contour may be
closed in the right half-plane without acquiring any residues.
In the final equality $C$ is a contour in the $\omega$ plane traversed
from $-i\infty$ to $+i\infty$ which crosses the real line 
within the strip $-(D-1) < \Re \omega < 0$ -- see 
Fig.~\ref{fig:scalarContours}.
By changing the cut prescription $i \epsilon \to i \epsilon s$ we deduce
\eq{ \label{eq:0}
  - \int_C \frac{d\omega}{2\pi i} 
  (2\omega+D-1) \Delta_\omega(Z-i\epsilon s)
  = 0 ,
}
which vanishes as a distribution. 

\begin{figure}
  \centering
  \includegraphics[width=0.3\textwidth]{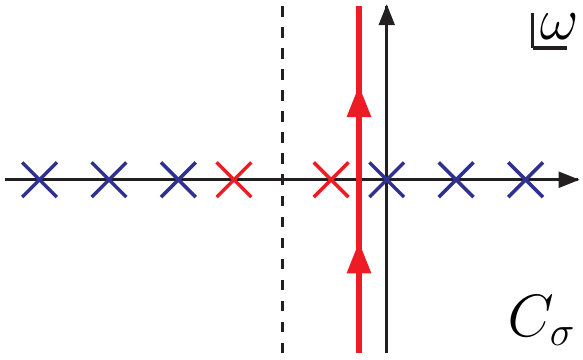}%
  \hspace{0.5cm}%
  \includegraphics[width=0.3\textwidth]{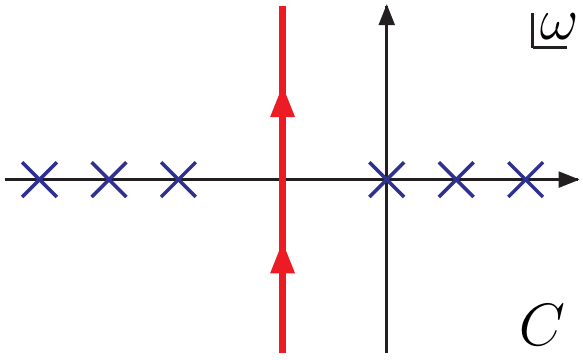}%
  \hspace{0.5cm}%
  \includegraphics[width=0.3\textwidth]{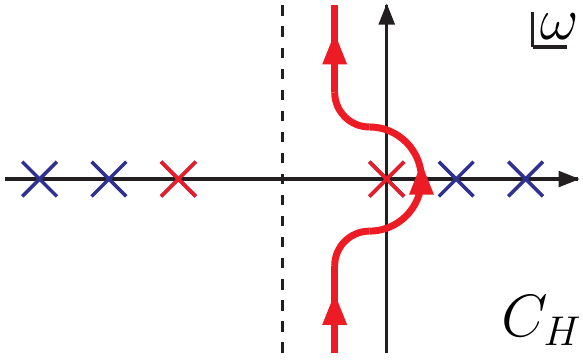}
  \caption{Integration contours in the complex $\omega$ plane.
    For each plot the dashed line denotes $\Re \omega = -(D-1)/2$,
    blue crosses denote the locations of singularities 
    in $\D_\omega(Z)$, and red crosses denote locations of other 
    singularities of the integrand.
    \label{fig:scalarContours}
  }
\end{figure}

It is well-known that there does not exist a maximally-symmetric solution
to the massless Klein-Gordon equation \cite{Allen:1987tz}.
This is due to the factor $\Gamma(\s)$ in (\ref{eq:Delta}): in the limit 
$M^2 \to 0$
\eq{
  \Delta_\sigma(Z) = \frac{1}{M^2 \vol} + \cO(M^0) , \quad M^2 \ll 1 ,
}
where $\vol = 2 \pi^{(D+1)/2}/\Gamma\left(\frac{D+1}{2}\right)$ 
is the volume of a unit $S^D$. Thus $\Delta_\s(Z)$ is undefined for $M^2 = 0$. 
Consider instead the maximally-symmetric solution to the equation 
\cite{Bros:2010aa}
\eq{
  (\Box - M^2) H_\sigma(Z-i\epsilon s) 
  = \frac{1}{\vol} .
}
The solution to this equation
is simply $H_\sigma(Z-i\epsilon s) = \Delta_\sigma(Z-i\epsilon s) 
- \frac{1}{M^2\vol}$.
Unlike $\Delta_\sigma(Z)$ the function $H_\s(Z)$ is regular in the 
neighborhood of $\sigma=0$, so $H_0(Z)$ exists and satisfies
\eq{
  \Box H_0(Z-i\epsilon s) = \frac{1}{\vol} .
}
Since $H_\sigma(Z)$ differs from $\D_\s(Z)$ only by an 
additive constant it is also a Hadamard distribution, and remains
so in the $\sigma \to 0$ limit.
$H_0(Z)$ may also be written as the limit
\eq{ \label{eq:Hlimit}
  H_0(Z) = \left[\frac{\d}{\d(M^2)} 
    \left( M^2 \Delta_\sigma(Z) \right) \right]_{M^2=0} ,
}
or may be written as a contour integral
\eq{ \label{eq:Hcontour}
  H_0(Z) = \int_{C_H} \frac{d\omega}{2\pi i} 
  \frac{(2\omega+D-1)}{\omega(\omega+D-1)}
  \Delta_\omega(Z) ,
}
where $C_H$ is a contour in the complex $\omega$ plane that is traversed
from $-i\infty$ to $+i\infty$ mostly in the left half-plane, but is
deformed so as to keep the pole at $\omega = 0$ to the left of the contour
-- see Fig.~\ref{fig:scalarContours}. To verify this expression, let
us act with the d'Alembertian:
\eqn{  \label{eq:boxH}
  \Box H_0(Z-i\epsilon s) 
  &=& \int_{C_H} \frac{d\omega}{2\pi i} 
  \frac{(2\omega+D-1)}{\omega(\omega+D-1)}  \Box \Delta_\omega(Z-i\epsilon s) 
  \nn \\
  &=& - \int_{C_H} \frac{d\omega}{2\pi i} 
  (2\omega+D-1) \Delta_\omega(Z-i\epsilon s) 
  \nn \\
  &=& - \int_C \frac{d\omega}{2\pi i} 
  (2\omega+D-1) \Delta_\omega(Z-i\epsilon s) 
  + \frac{1}{\vol}
  \nn \\
  &=& + \frac{1}{\vol} .
}
To obtain the third equality we deform the contour from $C_H$ to $C$;
along the way we pick up a residue from the simple pole at $\omega=0$
contained in $\Delta_\omega(Z)$.
The final equality follows from (\ref{eq:0}).
By changing the $i\epsilon$ prescription we obtain a distribution
which satisfies $\Box H_0(Z-i\epsilon) = \frac{1}{\vol} 
+ i \delta^D(x,\bx) / \sqrt{-g}$.

\subsection{Transverse traceless 2-pt function}
\label{sec:TT}

We are finally ready to combine our ingredients to compute
the 2-pt function $\CO{h_{\mu\nu}(x) h^{\bmu\bnu}(\bx)}$. 
In TT gauge this object satisfies the equation of motion
\eq{ \label{eq:TTEOM}
  \int_x \int_\bx f^{\mu\nu}(x) p_{\bmu\bnu}(\bx)
  \half \left(\Box - 2\right) 
  \C{ h_{\mu\nu}(x)h^{\bmu\bnu}(\bx)}_\Omega
  = 0 , \quad f,p \in \testTT .
}
The most general \emph{manifestly} maximally-symmetric ansatz 
for $\CO{h_{\mu\nu}(x) h^{\bmu\bnu}(\bx)}$ may be written
\eqn{ \label{eq:ansatz}
  \CO{h_{\mu\nu}(x) h^{\bmu\bnu}(\bx)}
  &=&
  \Delta_{\mu\nu}^{{\rm TT}\,\bmu\bnu}(Z - i \epsilon s) 
  \nn \\
  &:=&  \P_{\mu\nu}^{\phantom{\mu\nu}\a\b} 
  \P^{\bmu\bnu}_{\phantom{\bmu\bnu}\ba\bb} \left[
    A(Z-i \epsilon s) \T2{\a}{\b}{\ba}{\bb}
    \right] .
}
The term in brackets is consistent with the most general symmetric 
rank-2 MSBT (recall (\ref{eq:sly2}) and the fact that the projection
operators annihilate grad and metric terms).
This ansatz could be relaxed in several ways and still yield a 
maximally-symmetric state, but (\ref{eq:ansatz}) is the simplest and
will be sufficient for our purposes.

It is simple to obtain the equation of motion for $A(Z-i\epsilon s)$ 
from (\ref{eq:TTEOM}). First note that because the 
d'Alembertian commutes with the projection operators
\eq{ \label{eq:A1}
  \half \left(\Box - 2\right) 
  \Delta_{\mu\nu}^{{\rm TT}\,\bmu\bnu}(Z - i \epsilon s) 
  = \P_{\mu\nu}^{\phantom{\mu\nu}\a\b} 
  \P^{\bmu\bnu}_{\phantom{\bmu\bnu}\ba\bb} 
  \half \left(\Box - 2\right) \left[
    A(Z-i \epsilon s) \T2{\a}{\b}{\ba}{\bb}
  \right] .
}
Noting that for any bi-scalar $F(x,\bx)$
\eqn{ \label{eq:BoxOnSeed}
  \Box \left[ F(x,\bx) \T2{\a}{\b}{\ba}{\bb} \right] 
  &=& \left[(\Box+2) F(x,\bx) \right] \T2{\a}{\b}{\ba}{\bb} 
  \nn \\ & &
  + \grad + \met ,
}
we simplify (\ref{eq:A1}) to
\eqn{
  \half \left(\Box - 2\right) 
  \Delta_{\mu\nu}^{{\rm TT}\,\bmu\bnu}(Z - i \epsilon s) 
  &=& \P_{\mu\nu}^{\phantom{\mu\nu}\a\b} 
  \P^{\bmu\bnu}_{\phantom{\bmu\bnu}\ba\bb} 
  \left[ \half \Box A(Z- i \epsilon s) \right] \T2{\a}{\b}{\ba}{\bb} .
  \nn \\ & & 
}
Upon inserting this expression into (\ref{eq:TTEOM}), we 
use the self-adjointness of the projection operators (\ref{eq:Padjoint}),
as well as their action on TT tensors (\ref{eq:PonTT}), to obtain
\eqn{
  0 &=& \int_x \int_\bx f^{\mu\nu}(x) p_{\bmu\bnu}(\bx)
  \bigg\{
    \left[\frac{(D-3)^2}{8(D-2)^2} \Box^2 (\Box-(D-2))
      \bBox (\bBox-(D-2)) A(Z-i\epsilon s) \right] 
    \nn \\ & & \ph{\int_x \int_\bx f^{\mu\nu}(x) p_{\bmu\bnu}(\bx)
  \bigg\{}
      \T2{\mu}{\nu}{\bmu}{\bnu}
  \bigg\} 
}
In order for the right-hand side of this equality to vanish the
term in curly braces must be composed solely of grad and metric terms.
The most general such tensor with the index structure of 
$\T2{\mu}{\nu}{\bmu}{\bnu}$ is $(a + b Z) \T2{\mu}{\nu}{\bmu}{\bnu}$
with $a$ and $b$ arbitrary constants. Therefore we obtain the equation
of motion
\eq{ \label{eq:AEOM}
  \frac{(D-3)^2}{8(D-2)^2} \Box^3 (\Box-(D-2))^2 A(Z-i\epsilon s)
  = a + b Z ,
}
for arbitrary $a,b$.

We now show that there exists a maximally-symmetric solution 
to (\ref{eq:AEOM}) for $a = 1/\vol$ and for arbitrary $b$.
For the moment consider the case $b=0$; then the maximally-symmetric
solution is
\eq{ \label{eq:Acontour}
  A(Z-i\epsilon s) = \frac{8(D-2)^2}{(D-3)^2} \int_{C_A} \frac{d\omega}{2\pi i} 
  \frac{(2\omega+D-1)}{\omega^3(\omega+D-1)^3(\omega+1)^2(\omega+D-2)^2}
  \Delta_\omega(Z-i\epsilon s) .
}
where the integration contour $C_A:=C_H$ is the integration contour
used in (\ref{eq:Hcontour}) -- see also Fig.~\ref{fig:gravitonContours}.
The denominator of the integrand is simply the eigenvalue of 
$\D_\omega(Z-i\epsilon s)$ with respect to the derivative operator
on the left-hand side of (\ref{eq:AEOM}).
Let us explicitly verify that this is a solution to (\ref{eq:AEOM})
with $a=1/\vol$:
\eqn{
   & & \frac{(D-3)^2}{8(D-2)^2} \Box^3 (\Box-(D-2))^2 A(Z-i\epsilon s)
   \nn \\
   &=& \int_{C_A} \frac{d\omega}{2\pi i} 
  \frac{(2\omega+D-1)}{\omega^3(\omega+D-1)^3(\omega+1)^2(\omega+D-2)^2}
  \Box^3 (\Box-(D-2))^2 \Delta_\omega(Z-i\epsilon s) 
  \nn \\
   &=& - \int_{C_A} \frac{d\omega}{2\pi i} 
  (2\omega+D-1) \Delta_\omega(Z-i\epsilon s) 
  \nn \\
  &=& 
  {\rm Res}\left[ (2\omega+D-1) \Delta_\omega(Z-i\epsilon s) \right]_{\omega=0}
  - \int_{C} \frac{d\omega}{2\pi i} 
  (2\omega+D-1) \Delta_\omega(Z-i\epsilon s) 
  \nn \\
  &=& \frac{1}{\vol} .
}
To obtain the third equality we deform the integration contour from
$C_A$ to $C$ acquiring the residue of the pole at $\omega=0$ 
contained in $\D_\omega(Z)$ along the way. To obtain the final
equality we insert the value of this residue $1/\vol$ and
note that the remaining contour integral vanishes as a distribution
(recall (\ref{eq:0})).

We may also construct solutions to (\ref{eq:AEOM}) with $a=1/\vol$ and 
$b\neq 0$ by adding to (\ref{eq:Acontour}) a solution to (\ref{eq:AEOM})
with $a=0$ and $b\neq 0$.
Recalling that $\Box Z = - D Z$ we see that the solution to this equation 
is proportional to $Z$. However, the projection operators which act
on $A(Z-i\epsilon s)$ annihilate the term $Z \T2{\mu}{\nu}{\bmu}{\bnu}$ so this
additional term does not contribute to 
$\D^{{\rm TT}\bmu\bnu}_{\mu\nu}(Z-i\epsilon s)$.
Therefore the parameter $b$ represents a redundancy of description
introduced by our use of the projection operators; different choices of
$b$ yield the same 2-pt function $\Delta^{{\rm TT}\bmu\bnu}_{\mu\nu}(Z-i\epsilon s)$. 
Similarly, we may add to $A(Z-i\epsilon s)$ an arbitrary constant 
as the projection operators annihilate the term $\T2{\mu}{\nu}{\bmu}{\bnu}$. 
For simplicity we keep (\ref{eq:Acontour}) as our expression for 
$A(Z-i\epsilon s)$.\footnote{If we wish, we may also express $A(Z)$ 
  as a sum of terms of the form
  \eq{ \label{eq:Aseries}
    A(Z) = 
    \sum_{i=0}^4 c_i \left[\left(\frac{\d}{\d\s}\right)^i (M^2(\s) \D_\s(Z))\right]_{\s=0}
    + \sum_{j=0}^2 k_j \left[\left(\frac{\d}{\d\s}\right)^j \D_\s(Z)\right]_{\s=-1} .
  }
  To do so we deform the integration contour from $C_A$ to $C$, e.g. to $\Re \omega = -(D-1)/2$, along the way acquiring residues from the 
  poles at $\omega = 0$  and $\omega = -1$. These residues provide the
  terms (\ref{eq:Aseries}) and the remaining contour integral, 
  which is absolutely convergent, vanishes as the integrand is odd. 
  We will not need the coefficients in 
  (\ref{eq:Aseries}) so we do not bother to record them.
}

\begin{figure}
  \centering
  \includegraphics[width=0.3\textwidth]{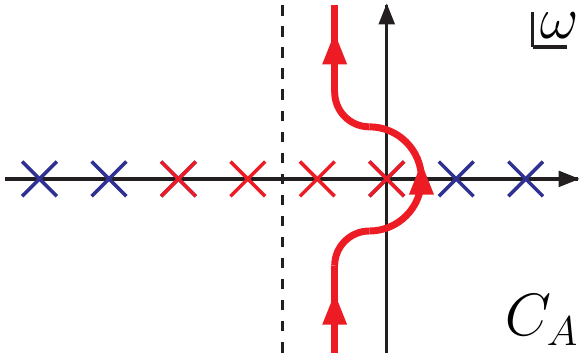}%
  \hspace{0.5cm}%
  \includegraphics[width=0.3\textwidth]{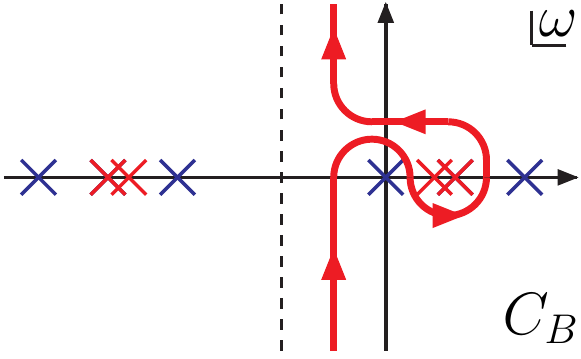}%
  \caption{Integration contours in complex $\omega$ plane.
    For each plot the dashed line denotes $\Re \omega = -(D-1)/2$,
    blue crosses denote the locations of singularities 
    in $\D_\omega(Z)$, and red crosses denote locations of other 
    singularities of the integrand.
    {\bf Left:} the contour $C_A=C_H$ utilized for $A(Z)$ and $U(Z)$. 
    {\bf Right:} the contour utilized for $B(Z)$ in
    the de Donder gauge term (see \S\ref{sec:deDonder}). The
    location of the poles at $\omega = \s_\b$, $-(\s_\b+D-1)$ depicted 
    correspond to transverse gauge $\beta=0$.
    \label{fig:gravitonContours}
  }
\end{figure}

The graviton 2-pt function simplifies considerably when smeared
against test functions and using this expression it is easy to 
verify that $\Delta^{{\rm TT}\bmu\bnu}_{\mu\nu}(Z-i\epsilon s)$ is 
positive and Hadamard. For $f,p \in \testTT$:
\eqn{ \label{eq:TTsmearing}
  \CO{h(f) h(p)} 
  &=& \int_x \int_\bx f^{\mu\nu}(x) p_{\bmu\bnu}(\bx)
  \Delta^{{\rm TT}\bmu\bnu}_{\mu\nu}(Z-i\epsilon s)
  \nn \\
  &=& 2 \int_x \int_\bx f^{\mu\nu}(x) p_{\bmu\bnu}(\bx)
  H_0(Z-i\epsilon s) \T2{\mu}{\nu}{\bmu}{\bnu} . \quad\quad
}
Here we have again utilized (\ref{eq:Padjoint}) and (\ref{eq:PonTT})
and we have also recognized $H_0(Z)$ (\ref{eq:Hcontour}).
By noting that
\eqn{
  \T2{\mu}{\nu}{\bmu}{\bnu} 
  &=& \frac{1}{4} 
  \nabla_\mu\nabla_\nu\nabla^\bmu\nabla^\bnu Z^2 + \met
  \nn \\
  &=& \frac{1}{4} \eta_{AC} \eta_{BD} 
  (\nabla_\mu\nabla_\nu X^A X^B)(\nabla^\bmu\nabla^\bnu \bX^A \bX^B)
  + \met ,\nn \\
}
where $\eta_{AB} = {\rm diag}\{-,+,\dots,+\}$ is the metric of the Minkowski 
embedding space,
we may use the chain rule to recast (\ref{eq:TTsmearing}) as
\eq{ \label{eq:TTsmearing2}
  \CO{h(f) h(p)} 
  = \half \eta_{AC} \eta_{BD} \int_x \int_\bx F^{AB}(x) P^{CD}(\bx)
  H_0(Z-i\epsilon s) 
}
with smearing functions
\eq{
  F^{AB}(x) := \nabla_\mu\nabla_\nu (f^{\mu\nu}(x) X^A X^B), 
  \quad
  P^{CD}(\bx) := \nabla^\bmu\nabla^\bnu (p_{\mu\nu}(x) \bX^C \bX^D) .
}
These test functions belong to the class
\eq{
  \mathscr{T}_{\rm S} := \left\{ f(x) \in C^\infty_0(dS_D) \; \bigg| \;
    \int_x f(x) = 0 \right\} .
}
Bros et.~al \cite{Bros:2010aa} have shown that $H_0(Z)$ is a positive 
kernel for this class of scalar test functions, so it follows 
that $\Delta^{{\rm TT}\bmu\bnu}_{\mu\nu}(Z-i\epsilon s)$ 
is a positive kernel for the test functions $\testTT$.
That $\D^{{\rm TT}\bmu\bnu}_{\mu\nu}(Z-i\epsilon s)$ is Hadamard
follows from the fact that $H_0(Z-i\epsilon s)$ is Hadamard.

The normalization of our 2-pt function has been fixed by examining 
the time-ordered correlation function $\CO{T h_{\mu\nu}(x) h^{\bmu\bnu}(\bx)}$
which is obtained by
changing the cut prescription:
$\CO{T h_{\mu\nu}(x) h^{\bmu\bnu}(\bx)} 
= \Delta^{{\rm TT}\bmu\bnu}_{\mu\nu}(Z-i \epsilon)$.
This object satisfies
\eq{ \label{eq:TorderedEOM}
  \int_x \int_\bx f^{\mu\nu}(x) p_{\bmu\bnu}(\bx)
  \half \left(\Box - 2\right) 
  \C{T h_{\mu\nu}(x)h^{\bmu\bnu}(\bx)}_\Omega
  = i \int_x f^{\mu\nu}(x) p_{\mu\nu}(x) , \quad f,p \in \testTT ,
}
or equivalently, the distribution equation
\eq{ \label{eq:TorderedEOM2}
  \half \left(\Box - 2\right) 
  \C{T h_{\mu\nu}(x)h^{\bmu\bnu}(\bx)}_\Omega 
  = i \delta^{{\rm TT}\bmu\bnu}_{\mu\nu}(x,\bx) ,
}
where $\delta^{{\rm TT}\bmu\bnu}_{\mu\nu}(x,\bx)$ is the identity
operator on the space of test functions $\testTT$. 
The identity operator satisfies
\eq{ \label{eq:deltaTT1}
  \int_x\int_\bx f^{\mu\nu}(x) p_{\bmu\bnu}(\bx)
  \delta^{{\rm TT}\bmu\bnu}_{\mu\nu}(x,\bx)
  = \int_x f^{\mu\nu}(x) p_{\mu\nu}(x) ,
  \quad f,p \in \testTT,
}
as well as
\eq{ \label{eq:deltaTT2}
  f_{\mu\nu}(x)
  =
  \int_\bx \delta^{{\rm TT}\bmu\bnu}_{\mu\nu}(x,\bx) q_{\bmu\bnu}(\bx)
  \in \testTT , \quad q \in C_0^\infty(T^0_2(dS_D)) .
}
As defined by its action on test function 
$\delta^{{\rm TT}\bmu\bnu}_{\mu\nu}(x,\bx)$
is maximally-symmetric, though it need not be written manifestly so.
Consider
\begin{proposition}
  The TT identity operator $\delta^{{\rm TT}\bmu\bnu}_{\mu\nu}(x,\bx)$
  which satisfies (\ref{eq:deltaTT1}) and (\ref{eq:deltaTT2})
  may be written
  {\normalfont $\ph{a}$
  \eq{
    \delta^{{\rm TT}\bmu\bnu}_{\mu\nu}(x,\bx)
    = \P_{\mu\nu}^{\phantom{\mu\nu}\a\b} 
    \P^{\bmu\bnu}_{\phantom{\bmu\bnu}\ba\bb} \left[ U(x,\bx) 
      \T2{\a}{\b}{\ba}{\bb} \right],
  }
  }
  where $U(x,\bx)$ is any solution to the equation
  {\normalfont $\ph{a}$
  \eq{ \label{eq:UEOM}
    \frac{(D-3)^2}{4(D-2)^2} \Box (\Box-(D-2))\bBox (\bBox-(D-2)) 
    U(x,\bx)
    = a + b Z + i \frac{\delta^D(x,\bx)}{\sqrt{-g}},
  }
  }
  with arbitrary coefficients $a$ and $b$.
\end{proposition}
The maximally-symmetric solution for $U(x,\bx)$
which satisfies (\ref{eq:UEOM}) with $a = 1/\vol,\,b=0,$ is 
given by
\eq{ \label{eq:U}
  U(Z-i\epsilon) = - \frac{4(D-2)^2}{(D-3)^2}
  \int_{C_A} \frac{d\omega}{2\pi i} 
  \frac{(2\omega+D-1)}{\omega^2(\omega+D-1)^2(\omega+1)^2(\omega+D-2)^2}
  \Delta_\omega(Z-i\epsilon) .
}
We obtain the normalization for the 2-pt function from this expression.

\subsection{Comparison with previous works}
\label{sec:compare}

We now discuss two points of contact between our result for the
graviton 2-pt function and existing results in the literature.
First, we note that our expression for the TT part of the 
2-pt function agrees with result obtained by
first constructing the 2-pt function on the Euclidean sphere $S^D$,
then analytically continuing the result to Lorentz-signature.
We show this in Appendix~\ref{app:Euclidean}.
This procedure has been used previously by many authors
\cite{Allen:1986dd,Floratos:1987aa,Higuchi:2000ge,Higuchi:2001uv,Park:2008ki},
but its validity for theories with massless fields 
has been debated \cite{Miao:2009hb,Miao:2010vs}.\footnote{For 
  interacting massive scalar 
  QFTs there is no debate: this procedure is been shown to agree 
  with Lorentz-signature constructions to all orders in perturbation
  theory \cite{Higuchi:2010aa,Korai:2012fi}.
  The massiveness of the fields is essential for the proof.
}
Since we have explicitly verified that our state is positive and
Hadamard it seems that the concerns of \cite{Miao:2009hb,Miao:2010vs} 
are not realized in this case.
One could have guessed this would be the case from the fact that the
Euclidean action for gravitons in TT gauge is bounded below, and
therefore the de Sitter version of the Osterwalder-Schrader 
theorem (see e.g. \cite{Schlingemann:1999mk}) would seem to assure
that the Lorentzian state defined by the analytically continued
Euclidean 2-pt function is free of pathologies (e.g., the Lorentzian
state is positive).

At first sight our results appear to be in conflict with those
\cite{Miao:2011fc}. This work claims that there does not exist a 
maximally-symmetric Hadamard solution to the TT part of the graviton 
2-pt Schwinger-Dyson equation.
The source of this tension is the fact that these authors
impose an additional restriction on the form of the 2-pt function:
they require that it admit a fourier transform with 
respect to the spatial coordinates in a Poincar\'e chart (\ref{eq:gPoincare})
which is convergent (in the sense of a function) in the limit where 
the momentum $\vec{k}\to 0$.
This requirement clearly forbids the possibility of any 
maximally-symmetric state whose 2-pt function increases in magnitude 
as $Z \to -\infty$ (infinite spatial separation), such as the
solution obtained above. As a result of this requirement on the
fourier transform \cite{Miao:2011fc} is required to consider a less 
symmetric ansatz for the function we call $A(Z)$ above.

We do not impose the restriction of \cite{Miao:2011fc} on fourier
transforms of the graviton 2-pt function, nor do we believe this restriction
is necessary to define a consistent quantum theory.
Obviously, this condition is unnatural if one considers the theory
on the global chart, where to impose this condition would require
specifying a preferred Poincar\'e chart, i.e. a preferred timelike direction.
Within the chosen chart this condition restricts the 
admissible class of gauge transformations 
$h_{\mu\nu}(x) \to h_{\mu\nu}(x) + 2\nabla_{(\mu}\xi_{\nu)}(x)$ 
to those for which $\nabla_{(\mu}\xi_{\nu)}(x)$ vanishes sufficiently rapidly 
near the spatial boundary.
As stated in \S\ref{sec:gauge}, we define our theory of linearized gravity
such that all smooth $\xi_\mu(x)$ may generate gauge transformations.

In the remainder of this section we analyse the relationship between
the maximally-symmetric state $\Omega$ constructed above and the state 
defined by the less-symmetric 2-pt function of \cite{Miao:2011fc}.
We will shortly conclude that within our framework the two states are
equivalent. This discussion is very technical and may not be of
interest to all readers.

Ref.~\cite{Miao:2011fc} utilizes the same ansatz 
(\ref{eq:ansatz}) but with $A(Z)$ replaced by a function
$\mathcal S_2(x,\bx)$ which is not assumed to be
maximally-symmetric.\footnote{Ref.~\cite{Miao:2011fc} uses
  the distance measure $y(x,\bx) = 2(1-Z(x,\bx))$, but for clarity
  we will continue to use $Z(x,\bx)$.}
These authors obtain the same equation of motion for $\mathcal S_2(x,\bx)$
(\ref{eq:AEOM}); however, they only consider the values $a=0,\, b=0$,
and as a result do not obtain the maximally-symmetric solution that
exists for $a = 1/\vol$. 
The solution of \cite{Miao:2011fc} is explored in detail in 
\cite{Kahya:2011sy} whose notation we now follow. 
The form of the solution is 
$\mathcal S_2(x,\bx) = S_2(Z) + \delta S_2(x,\bx)$,
where $S_2(Z) = A(Z)$ and $\delta S_2(x,\bx)$ is a symmetry-breaking
term which is a solution
to (\ref{eq:AEOM}) with $a = -1/\vol,\; b=0$. The authors choose
$\delta S_2(x,\bx)$ so as be invariant under spatial translations 
and rotations in the Poincar\'e chart; the explicit form of 
$\delta S_2(x,\bx)$ is given by eq.~(84) of \cite{Kahya:2011sy}.
Ref.~\cite{Kahya:2011sy} also showed that $\delta S_2(x,\bx)$ is not
annihilated by the projection operators in the ansatz (\ref{eq:ansatz}).

Let us denote the state defined by the 2-pt function of 
\cite{Miao:2011fc} by MTW (after the authors). To see that 
MTW is equivalent to $\Omega$ examine
the smeared 2-pt function $\C{h(f) h(p)}_{\rm MTW}$. Inserting
$\mathcal S_2(x,\bx)$ into this expression and utilizing 
the self-adjointness of the projection operators (\ref{eq:Padjoint})
as well as their action on TT tensors (\ref{eq:PonTT}) we
obtain
\eq{\label{eq:MTWsmearing1}
  \C{h(f) h(p)}_{\rm MTW} = \int_x \int_\bx f^{\mu\nu}(x) p_{\bmu\bnu}(\bx)
    \left[ H_0(Z-i\epsilon s) + Q(x,\bx) \right] \T2{\mu}{\nu}{\bmu}{\bnu}  ,
}
where $H(Z)$ is as in (\ref{eq:TTsmearing2}) and
\eqn{
  Q(x,\bx) &:=& \frac{(D-3)^2}{8(D-2)^2} \Box (\Box-(D-2)) \bBox (\bBox-(D-2))
  \delta S_2(x,\bx)
  \nn \\
  &=& - \frac{1}{(D-1)\vol} \left[ \ln \eta + \ln \etab\right] .
}
Here $\eta$ is the time coordinate in the Poincar\'e chart 
(\ref{eq:gPoincare}). This function satisfies
\eq{
  \Box Q(x,\bx) = \bBox Q(x,\bx) = -\frac{1}{\vol} ,
}
and as a result the term in brackets in (\ref{eq:MTWsmearing1})
is a Hadamard bi-solution to the massless Klein-Gordon equation.
By judicious use of the chain rule we show that $Q(x,\bx)$
does not contribute to (\ref{eq:MTWsmearing1}):
\eqn{
  \C{h(f) h(p)}_{\rm MTW}  - \CO{h(f) h(p)}
  &=& \int_x \int_\bx f^{\mu\nu}(x) p_{\bmu\bnu}(\bx) 
  Q(x,\bx) \T2{\mu}{\nu}{\bmu}{\bnu}
  \nn \\
  &=& \int_x \int_\bx f^{\mu\nu}(x) p_{\bmu\bnu}(\bx) 
   \left[ \nabla_{(\mu} \nabla^{(\bmu}  Q(x,\bx) \right] Z
   (\nabla_{\nu)}\nabla^{\bnu)} Z )
  \nn \\
  &=& 0 .
}
The final equality follows from the fact that
$\nabla_{\mu} \nabla^{\bmu}  Q(x,\bx) = 0$.
This computation shows that the states MTW and $\Omega$ are
equivalent as probed by \emph{all} local observables 
$\cA(dS_D)$ and thus are equivalent states within the 
algebraic framework.\footnote{Ref.~\cite{Mora:2012zh} has shown the
  the lesser condition that $\delta S_2(x,\bx)$ does not 
  contribute to the linearized Weyl tensor 2-pt 
function.}
This computation also shows that the state MTW admits a globally
Hadamard extension from the algebra of observables of the Poincar\'e
chart to the algebra of observables of the global chart.\footnote{
A related computation one could perform
is to check that 
$\C{T h_{\mu\nu}(x) h^{\bmu\bnu}(\bx)}_{\rm MTW}  
- \CO{T h_{\mu\nu}(x) h^{\bmu\bnu}(\bx)}$ is a homogeneous
solution to the equations of motion. Indeed
\eq{
  -\half(\Box-2)\left[ \C{T h_{\mu\nu}(x) h^{\bmu\bnu}(\bx)}_{\rm MTW}  
    - \CO{T h_{\mu\nu}(x) h^{\bmu\bnu}(\bx)}\right]
  = 0
}
holds as a distribution equation; that is, the left hand side vanishes
as tested by any functions in $\testTT$. }


\section{Consequences}
\label{sec:consequences}

The previous section examined the maximally-symmetric state $\Omega$
in great detail; in this section we use our results to establish
a few basic properties of \emph{generic} states.
Recall that we may connect the notion of an algebraic state with the 
more familiar notion of a vector in a Hilbert space via the
``GNS construction'': given an algebraic state $\Omega$ the GNS 
construction provides a Hilbert space $\cH_\Omega$ containing a cyclic
vector $\ket{\Omega}$, a representation of the abstract elements of 
the observable algebra (which we will not distinguish in notation),
as well as a dense set of state vectors $\cD \subset \cH_\Omega$ 
constructed by application of the observable algebra on $\ket{\Omega}$, i.e. 
\eq{
  \cD := \{ A \ket{\Omega} \;|\; A \in \cA(dS_D) \} .
}
For technical reasons in this section we restrict attention to 
observables whose support is contained in a \emph{contractible} region:
we call a compact subset $\cO \subset dS_D$ contractible if its boundary 
$\d \cO$ may be contracted to a point.
We denote the resulting algebra by $\cA_c(dS_D)$ and the dense set of states 
generated by $\cA_c(dS_D)$ on $\cH_\Omega$ by $\cD_c$:
\eq{
  \cD_c := \{ A \ket{\Omega} \;|\; A \in \cA_c(dS_D) \} .
}
Note that all compact co-dimension 1 surfaces 
in the Poincar\'e or static charts are contractible,
so for quantization on these backgrounds $\cA_c(dS_D) = \cA(dS_D)$
and $\cD_c = \cD$;
it is only on the global chart that $\cA_c(dS_D) \subset \cA(dS_D)$,
$\cD_c \subset \cD$.
We should emphasize that the set of states $\cD_c$ can approximate
\emph{any} state on $\cH_\Omega$ and are not limited to 
highly-symmetric or ``static'' configurations.

We will use the following lemma of Higuchi frequently \cite{Higuchi:2012vy}:
\begin{lemma} \label{lem:Higuchi}
  Every $f \in \testTT$ whose support is contained in a contractible
  region may be written
  \eq{ \label{eq:Higuchi}
    f^\s_{\ph{\s}\nu}(x) 
    = \Upsilon^{[\a\b]\,[\rho\s]}_{[\gamma\delta]\,[\mu\nu]} 
    \left( \nabla_\rho\nabla^\mu 
      + \delta^\mu_\rho \right) v_{\a\b}^{\ph{\a\b}\gamma\delta}(x),
  }
  where $\Upsilon^{[\a\b]\,[\rho\s]}_{[\gamma\delta]\,[\mu\nu]}$ is as in
  (\ref{eq:Upsilon}) and $v \in C^\infty_0(T^2_2(dS_D))$ is a test 
  tensor with the (anti-)symmetries and tracelessness of the Weyl tensor.
  The support of $v$ may be made arbitrarily close to that of 
  $f$.
\end{lemma}
As pointed out by Higuchi \cite{Higuchi:2012vy}, it is an immediate 
consequence of this lemma that for every metric observable 
$h(f) \in \cA_c(dS_D)$ there exists an equivalent linearized Weyl tensor 
observable $C^{(1)}(v) \in \cA_c(dS_D)$ with $f$ and $v$ related as 
in (\ref{eq:Higuchi}).

\subsection{The Reeh-Schlieder theorem}
\label{sec:RS}

In this section we prove a version of the Reeh-Schlieder theorem
(theorem 5.3.1 of \cite{Haag:1992aa}) for linearized gravity on de Sitter
backgrounds:
\begin{theorem}\label{thm:RS}
  The set of states
  \eq{
    \cA_c(\cO)\Omega := \{ A \ket{\Omega} \; | \; A \in \cA_c(\cO) \} 
  }
  generated from operators in a contractible subset $\cO \subset dS_D$ 
  is dense on $\cD_c$. 
\end{theorem}
The Reeh-Schlieder theorem has many interesting consequences 
\cite{Haag:1992aa,Streater:1989vi}, 
but our immediate purpose for introducing this theorem is
that it allows us to describe any state in $\cD_c$ by a state in 
$\cA_c(\cO)\Omega$,
a feature we will need in \S\ref{sec:hair}.
We note that for quantum fields of spin less than 2
the Reeh-Schlieder theorem has been proven, under various circumstances,
on a number of curved backgrounds -- see e.g. \cite{Schlieder:1965aa,Borchers:1998nw,Jaekel:2000tk,Strohmaier:2000aa,Strohmaier:2002aa,Sanders:2009aa,Dappiaggi:2011fk}.

The Reeh-Schlieder theorem is often viewed as a statement about 
the analyticity properties of the n-pt functions of state vectors 
in $\cD_c$.
For scalar QFTs on Minkowski space the Reeh-Schlieder theorem may be proven by
showing that n-pt functions are holomorphic on a sufficiently large
domain in complex Minkowski space such that 
determining their value on an open region of the real Lorentzian section
determines the n-pt functions everywhere on the section.
There is the subtlety that the domain of holomorphicity does not
actually contain the real Lorentzian section -- there n-pt functions are
distributions, i.e. ``boundary values of holomorphic functions'' -- 
but this is accommodated for by the ``edge-of-the-wedge'' theorem 
\cite{Streater:1989vi}.
For scalar QFTs on de Sitter there exists a rather straight-forward 
extension of this proof by Bros et.~al.~\cite{Bros:1998ik}
which uses the holomorphicity properties of scalar n-pt functions on
complexified de Sitter space (or equivalently, the complexified 
embedding space). In our expressions for the graviton 2-pt function 
obtained in \S\ref{sec:TT} the functions $A(Z)$ and $H_0(Z)$
have the same domain of holomorphicity as the 2-pt functions of
Klein-Gordon fields on $dS_D$,\footnote{In particular, they satisfy
the ``weak spectral condition'' of \cite{Bros:1998ik}. This 
is not surprising given that both $A(Z)$ and $H_0(Z)$ are 
closely related to the Klein-Gordon 2-pt function $\D_\sigma(Z)$.}
so it appears that it would be rather simple
to extend the proof of \cite{Bros:1998ik} to linearized gravity.
Ultimately we chose another path for our proof, but we mention
this tactic because it provides a nice geometric description.

A faster if less picturesque approach to proving the Reeh-Schlieder theorem 
is to employ the tools of microlocal analysis. In particular, the 
analyticity properties of a distribution may be characterized by its 
analytic wave front set ($\WFA$).\footnote{A mathematical introduction to 
  this technology may be found in Ch.~9 
  of \cite{Hormander:1990aa}. For introductions to the applications
  to QFT see \cite{Brunetti:1995rf,Verch:1999aa}.}
In this language, one expects the Reeh-Schlieder theorem to hold 
on a dense set of states whose correlation functions have analytic wave 
front sets that are ``sufficiently small,'' indicating a high degree of 
analyticity.
There exists a powerful microlocal version of the edge-of-the-wedge 
theorem, namely Proposition 5.3 of \cite{Strohmaier:2002aa}:
\begin{proposition}\label{prop:wedge}
  Let $M$ be a real analytic connected manifold and $u \in D'(M)$ 
  a distribution with the property that
  \eq{ \label{eq:wedge}
    \WFA(u) \cap - \WFA(u) = \emptyset. 
  }
  Then for each non-void open subset $\cO \subset M$
  if the restriction of $u$ to $\cO$ vanishes then $u = 0$.
\end{proposition}
Using this proposition \cite{Strohmaier:2002aa} proved the Reeh-Schlieder
theorem for scalar QFTs on real analytic spacetimes whose state
n-pt functions satisfy (\ref{eq:wedge}).
Below we show that the analytic wave front sets of linearized Weyl 
tensor n-pt functions of states in $\cD_c$ also satisfy (\ref{eq:wedge}). 
Following \cite{Strohmaier:2002aa} we may then quickly prove the 
Reeh-Schlieder theorem for linearized gravitons on $dS_D$.

\begin{lemma}\label{lem:WFAs}
  Consider two state vectors $\ket{\Psi},\;\ket{\Theta} \in \cD_c$.
  The analytic wave front set of the amplitudes
  \eq{\label{eq:Gamp}
    G^n_{\Psi\Theta} :=
    \bra{\Psi} 
    C^{(1)\a_1}_{\ph{(1)\a_1}\b_1 \gamma_1 \delta_1}(x_1)\dots
    C^{(1)\a_n}_{\ph{(1)\a_n}\b_n\gamma_n \delta_n}(x_n) 
    \ket{\Theta} 
  }
  satisfy (\ref{eq:wedge}).
\end{lemma}

\begin{proof}
We begin by determining the analytic wave front set of the
graviton n-pt functions of $\Omega$:
\eq{
  F^n_{\Omega\Omega} 
  := \CO{h_{\mu_1 \nu_1}(x_1)\dots h_{\mu_n \nu_n}(x_n)} .
}
Since $\Omega$ is quasi-free 
$F^n_{\Omega\Omega}$ vanishes for $n$ odd while for $n$ even
\eq{ \label{eq:Feven}
  F^n_{\Omega\Omega}
  = \sum_{P} \prod_{r\in P}
  \CO{ h_{\mu_{r1}\nu_{r1}}(r_1) h_{\mu_{r2}\nu_{r2}}(r_2) } .
}
Here $P$ denotes the a partition of the set $\{x_1,\dots,x_n\}$
into pairs of points labelled $r = (r_1, r_2)$; the points are
labelled so as to preserve the Wightman operator ordering.
The analytic wave front set of the 2-pt function is
easily determined from that of $H_0(Z)$:
\eqn{
  & & \Xi(x_1,k_1 ; x_2, k_2) := \WFA\left( F^2_{\Omega\Omega} \right)
  \nn \\ & &= \{ (x_1, k_1 ; x_2, -k_2) \in (T^*dS_D \setminus \{0\})^2 \; | 
  \; (x_1,k_1)\sim (x_2,k_2),
  \; k_1 \in V_1^+ \} . \quad
}
Here we have once again used the notation that
$(x_1,k_1)\sim (x_2,k_2)$ denotes that $x_1$ and $x_2$ may
be joined by a null geodesic $k_1$ and $k_2$ are cotangent and
coparallel to that geodesic.
Then the analytic wave front set of $F^n_{\Omega\Omega}$ is 
\eq{
  \WFA\left(F^n_{\Omega\Omega}\right)
  = \cup_P \oplus_{r\in P} \Xi(x_{r1},k_{r1} ; x_{r2}, k_{r2}) .
}
Note that $\WFA\left(F^n_{\Omega\Omega}\right)$ satisfy (\ref{eq:wedge}).

Next we consider the linearized Weyl tensor n-pt functions $G^n_{\Omega\Omega}$
which may be constructed from the $F^n_{\Omega\Omega}$ by repeated use of
(\ref{eq:Weyl}).
Using the basic facts that i) $\WFA( \d u) \subseteq \WFA( u )$,
and ii) $\WFA(f u) \subseteq \WFA(u)$ for $f \in C^\infty$
we readily determine that
\eq{
  \WFA\left(G^n_{\Omega\Omega}\right)
  \subseteq \WFA\left(F^n_{\Omega\Omega}\right)
  = \cup_P \oplus_{r\in P} \Xi(x_{r1},k_{r1} ; x_{r2}, k_{r2}) .
}
Thus $\WFA\left(G^n_{\Omega\Omega}\right)$ satisfy (\ref{eq:wedge}).

States in $\cD_c$ are constructed by acting on $\Omega$ with
finite polynomials of the linearized Weyl tensor smeared by contractible 
$C_0^\infty$ test functions. Thus the amplitude $G^n_{\Psi\Theta}$ 
are a sum of a finite number of $G^m_{\Omega\Omega}$ with $n \le m < \infty$,
and it follows that $\WF( G^n_{\Psi\Theta})$ satisfy (\ref{eq:wedge}). $\blacksquare$
\end{proof}\\

\noindent {\it Proof of Theorem \ref{thm:RS}.}
    Consider the amplitudes $\bra{\Psi}A \ket{\Omega}$
    with $A \in \cA_c(\cO)$ and $\cO$ a contractible region on $dS_D$.
    Due to Lemma~\ref{lem:Higuchi} we may restrict attention to the
    linearized Weyl tensor correlators
    \eq{ \label{eq:RSC}
      \bra{\Psi} C^{(1)}(v_1)\dots C^{(1)}(v_n) \ket{\Omega}, 
      \quad {\rm supp} \,v_i \subseteq \cO .
    }
    Let $\ket{\Psi}$ be in the orthogonal complement to $\cA_c(\cO)\Omega$ so
    that (\ref{eq:RSC}) vanishes for all $v_i$.
    It follows that the restriction of each $G^n_{\Psi\Omega}$ to $\cO$ vanishes.
    Due to Lemma~\ref{lem:WFAs} we may use Proposition~\ref{prop:wedge} to 
    conclude that each $G^n_{\Psi\Omega} = 0$ on all of $(dS_D)^n$. Then all 
    correlation functions of $\Psi$ vanish and hence $\Psi = 0$, i.e. the 
    orthogonal complement of $\cA_c(\cO)\Omega$ is empty. $\blacksquare$

\subsection{Cosmic no-hair theorem}
\label{sec:hair}

We are almost ready to state the no-hair theorem. The remaining
ingredient we need is the cluster decomposition of $\Omega$ 
correlation functions of local observables.
Cluster decomposition does \emph{not} hold for the unsmeared graviton 
correlation function $\CO{h_{\mu\nu}(x) h^{\bmu\bnu}(\bx)}$.
Recall that we may write
\eq{ \label{eq:idunno}
  \CO{ h_{\mu\nu}(x) h^{\bmu\bnu}(\bx) }
  = H_0(Z) \T2{\mu}{\nu}{\bmu}{\bnu} + \grad + \met .
}
At large $|Z| \gg 1$ the function $H_0(Z)$ has the asymptotic
expansion
\eq{ \label{eq:HLargeZ}
  H_0(Z) = c_1 \ln (1+Z) + c_2 + \cO(Z^{-2})  , \quad
  |Z| \gg 1 ,
}
where $c_1,\;c_2$ are coefficients whose value is unimportant now.
The $c_2$ term is pure $\grad+\met$ (recall (\ref{eq:T2gauge})).
Thus for large $|Z| \gg 1$ we may write (\ref{eq:idunno}) as
\eq{ \label{eq:hhLargeZ}
  \CO{ h_{\mu\nu}(x) h^{\bmu\bnu}(\bx)} 
  = \left[ - c_1 \ln(1+Z) + \cO(Z^{-2}) \right] 
  \GG{\mu}{\nu}{\bmu}{\bnu} + \grad + \met ,
}
where $g_{\mu\bnu}$ is the parallel propagator (\ref{eq:pp}).
which is $\cO(Z^0)$ at large $|Z| \gg 1$. Therefore for $|Z| \gg 1$ the
2-pt function behaves like $\CO{ h_{\mu\nu}(x) h^{\bmu\bnu}(\bx)} \sim \ln Z$
as expected for a massless spin-2 field.

However, the $\Omega$ correlation functions of local observables
$A \in\cA_c(dS_D)$ \emph{do} enjoy a de Sitter version of cluster decomposition 
associated with large timelike and achronal separations \cite{Higuchi:2012vy}.
Consider the 2-pt function $\CO{h(f) h(p)}$ for two observables
that are well-separated in the sense that $|Z_{\rm min}| \gg 1$
where
\eq{ 
  Z_{\rm min} = {\rm min} \{ Z(x,\bx) \; | \;
  x \in {\rm supp}\, f^{\mu\nu}(x), \;
  \bx \in {\rm supp}\, p^{\bmu\bnu}(\bx) \}  .
}
Large $|Z| \gg 1$ corresponds to large timelike or achronal
separation. From Lemma~\ref{lem:Higuchi} this 2-pt function may
be recast as linearized Weyl tensor correlation function
$\CO{ C^{(1)}(v) C^{(1)}(u)}$ with appropriate $v, u$.
We may then use the known asymptotic behavior of the linearized
Weyl tensor 2-pt function to bound this 2-pt function
\cite{Kouris:2001aa,Mora:2012kr,Mora:2012zh}:\footnote{
  The most explicit expression available for the linearized Weyl 2-pt function 
  in the form (\ref{eq:WeylBound}) are given in the Corrigendum of 
  \cite{Kouris:2001aa} for $D=4$ dimensions.
  For other spacetime dimensions on may obtain (\ref{eq:WeylBound})
  from eq.~(94) of \cite{Mora:2012kr}.}
\begin{proposition}\label{prop:WeylBound}
  The linearized Weyl tensor 2-pt function of the state $\Omega$
  may be written
  \eq{ \label{eq:WeylBound}
    \CO{C^{(1)\a}_{\ph{(1)\a}\b\gamma\delta}(x)
      C^{(1)\ba}_{\ph{(1)\ba}\bb\bg\bd}(\bx)}
    = \sum_{i} F_i(Z) 
    T^{(i)\,\a\ph{\b\gamma\delta}\ba}_{\ph{(i)\,\a}\b\gamma\delta\ph{\ba}\bb\bg\bd}(Z) ,
  }
  with $F_i(Z)$ that satisfy
  \eq{
    F_i(Z) < \frac{c}{|Z|^2} , \quad {\rm for} \; |Z| \gg 1 , 
  }
  with $c$ a finite constant and tensors $T^{(i)}(Z)$ composed of 
  the parallel propagator and 
  normal vectors defined in (\ref{eq:pp}) and (\ref{eq:n}).
\end{proposition}
The parallel propagator and normal vectors are $\cO(Z^0)$ for 
$|Z| \gg 1$ and so it follows that we may bound
\eq{ \label{eq:clustering1}
  \CO{h(f) h(p)} < \frac{c}{|Z_{\rm min}|^2}, \quad
  |Z_{\rm min}| \gg 1 , \quad
  h(f), \; h(p) \in \cA_c(dS_D) ,
}
where $c$ is a finite constant that depends on the test functions
$f,p$ but not on $Z_{\rm min}$.

The cosmic no-hair theorem follows immediately from (\ref{eq:clustering1}):
\begin{theorem}\label{thm:hair}
  Let $\Psi \in \cA_c(\cO)\Omega $, $\cO$ a contractible subset of $dS_D$,
  and $A(f) \in \cA_c(dS_D)$ with $f$ schematically denoting the 
  test function. Then
  \eq{ \label{eq:hair}
    \left| \C{ A(f) }_\Psi - \CO{ A(f) } \right|
    < \frac{c}{|Z_{\rm min}|^2},
    \quad {\rm for}\; |Z_{\rm min}| \gg 1  ,
  }
  where
  \eq{\label{eq:Zmin}
    Z_{\rm min} = {\rm min} \{Z(x,\bx) \; | \;
    x \in {\rm supp}\, f(x), \; \bx \in \cO \} .
  }
  Here $c$ is a finite constant that depends on the test function
  $f$ and state $\Psi$ but does not depend on $Z_{\rm min}$.
\end{theorem}

\begin{proof}
Let $\ket{\Psi} = B \ket{\Omega}$, $B \in \cA_c(\cO)$. 
Like all states $\ket{\Psi}$ is normalized, so
$\bra{\Omega}B^*B \ket{\Omega} = 1$.
The correlation function 
$\C{A(f) }_\Psi= \bra{\Omega}B^* A(f) B \ket{\Omega}$
is constructed out of smeared products of the graviton 2-pt function.
From (\ref{eq:clustering1}) it follows that any such 2-pt
function connecting $A(f)$ and $B$ is $\cO(|Z_{\rm min}|^2)$. Thus
\eqn{
  \C{A(f) }_\Psi
  &=& \bra{\Omega}B^* A(f) B \ket{\Omega}  \nn \\
  &=& \bra{\Omega}B^*B \ket{\Omega} 
  \bra{\Omega}  A(f) \ket{\Omega} + \cO( |Z_{\rm min}|^{-2} ) \nn \\
  &=& \CO{ A(f) } + \cO( |Z_{\rm min}|^{-2} ) ,
}
which verifies (\ref{eq:hair}). $\blacksquare$
\end{proof}

The implications of this theorem are plain to see. Consider an
arbitrary state $\ket{\Psi} \in \cD_c$. Using the Reeh-Schlieder property
we may write this state as $\ket{\Psi} = B \ket{\Omega}$, $B \in \cA_c(\cO)$
with $\cO$ a contractible region centered around the point 
$(\eta_i,\vec{0})$ in the Poincar\'e chart (\ref{eq:gPoincare}).
Here $\eta_i$ is a finite but arbitrarily early time.
Now consider any observable $A \in \cA_c(dS_D)$ whose support is
centered around $(\eta_f,\vec{0})$ with $\eta_f \gg \eta_i$.
For sufficiently large $\eta_f-\eta_i$ the embedding distance $Z(x_i,x_f)$
between any point $x_f \in {\rm supp}(A)$ and any point $x_i \in \cO$ is 
large and Theorem~\ref{thm:hair} states that
\eq{\label{eq:nohair}
  |\C{A(\tau)}_\Psi - \CO{A(\tau)}| < c\, e^{-2\tau}, 
}
where $\tau$ is the proper time separation between $(\eta_i,\vec{0})$
and $(\eta_f,\vec{0})$. We have used the fact that $|Z| \sim e^{\tau}$ 
at large timelike separations. Obviously (\ref{eq:nohair}) remains
true if $A$ is displaced from the origin of the spatial slice so
long as it remains well within the causal future of $\cO$.

\section{de Donder 2-pt functions}
\label{sec:deDonder}

In this section we change gears and construct the graviton 2-pt
function of the state $\Omega$ in the class of generalized de Donder
(dD) gauges defined by the gauge condition (\ref{eq:dD}).
Our interest in these expressions stems from the fact that,
unlike the TT gauge condition used above, the dD gauge conditions 
may be imposed in non-linear perturbation theory, so the 2-pt
functions we construct here can be used as Green's functions
in perturbation theory.
We also take this opportunity to ease the tension in the 
literature regarding the existence of manifestly maximally-symmetric
2-pt functions for this useful class of gauges. 
As in the previous sections we use Lorentz-signature techniques.

In order to construct an ansatz for the 2-pt function let us first
note that any TT tensor trivially satisfies the dD gauge 
condition (\ref{eq:dD}) for all $\beta$. Therefore our ansatz for the 
dD 2-pt function contains the TT part obtained in \S\ref{sec:TT} as well 
as a second term:
\eq{ \label{eq:dDansatz1}
  \CO{ h_{\mu\nu}(x)h^{\bmu\bnu}(\bx)} = 
  \Delta^{{\rm TT}\bmu\bnu}_{\mu\nu}(Z-i\epsilon s ) 
  + \Delta^{{\rm dD}\bmu\bnu}_{\mu\nu}(Z-i\epsilon s) .
}
From the general form of an MSBT (\ref{eq:sly}) it follows that the 
dD term may be written
\eq{\label{eq:Cdef}
  \Delta^{{\rm dD}\bmu\bnu}_{\mu\nu}(Z) = C(Z) g_{\mu\nu}g^{\bmu\bnu}
  + \grad,
}
where the $\grad$ part is determined from $C(Z)$ by imposing the
gauge condition (\ref{eq:dD}). This is easily accomplished by using
the dD projection operator \cite{Mora:2012zi}
\eq{ \label{eq:b}
  \P^{(\beta)}_{\mu\nu} \phi(x) 
  := \left( \nabla_\mu \nabla_\nu + a_\b g_{\mu\nu} 
    + b_\b g_{\mu\nu} \Box \right) 
  \phi(x) 
  = \left( \nabla_\mu \nabla_\nu + b_\b g_{\mu\nu}(\Box - M^2_\beta) \right)
  \phi(x) ,
}
with
\eq{
  a_\beta = \frac{2(1-D)}{2-\beta D} = (1+D b_\b), \quad
  b_\beta = \frac{\beta-2}{2-\beta D} , \quad
  M_\beta^2 = - \frac{a_\b}{b_\b} = \frac{2(D-1)}{\beta-2} .
}
For any scalar field $\phi(x)$ the tensor $\P^{(\beta)}_{\mu\nu}\phi(x)$
satisfies (\ref{eq:dD}). Therefore our ansatz for 
$\Delta^{{\rm dD}\bmu\bnu}_{\mu\nu}(Z)$ is
\eq{ \label{eq:dDansatz2}
  \Delta^{{\rm dD}\bmu\bnu}_{\mu\nu}(Z) := \P_{\mu\nu} \P^{\bmu\bnu} B(Z) .
}

A rather awkward feature of the dD gauges is that in general the
linearized equations of motion do not preserve the gauge condition.
Recall that $L^{(1)}_{\mu\nu}(q)$ is transverse for any $q_{\mu\nu}(x)$, 
so the action of $L^{(1)}_{\mu\nu}$ 
preserves the gauge condition only for transverse gauge $\beta=0$.
Noting that
\eq{
  L^{(1)}_{\mu\nu}(\grad) = 0, \quad
  L^{(1)}_{\mu\nu}(g \phi) 
  = \frac{D-2}{2} \P^{(0)}_{\mu\nu} \phi(x) ,
}
where 
\eq{ \label{eq:Ptrans}
  \P_{\mu\nu}^{(0)} \phi(x)
  = \left(\nabla_\mu\nabla_\nu - g_{\mu\nu}(\Box+D-1)\right) \phi(x) 
}
constructs a transverse tensor from a scalar, we obtain
\eq{
  L_{\mu\nu}^{(1)}(\P^{(\beta)} \phi)
  = \frac{D-2}{2}b_\beta \P^{(0)}_{\mu\nu} 
  \left(\Box - M^2_\beta \right)\phi(x) .
}
So, in admittedly awkward notation, the action of the linearized
equations of motion on our ansatz is
\eq{ \label{eq:dDEOM1}
  L_{\mu\nu}^{(1)}(\D^{{\rm dD}\bmu\bnu})
  = \frac{D-2}{2}b_\beta \P^{(0)}_{\mu\nu} \P^{(\beta)\bmu\bnu} 
   \left[  \left(\Box - M^2_\beta \right) B(Z)  \right] .
}

We now proceed to obtain the equation of motion for $B(Z-i\epsilon s)$.
The graviton 2-pt function satisfies the equation of motion
\eq{ \label{eq:TEOM}
  \int_x \int_\bx f^{\mu\nu}(x) p_{\bmu\bnu}(\bx)
  \CO{ L^{(1)}_{\mu\nu}(h) h^{\bmu\bnu}(\bx) } = 0 ,
  \quad  p \in \testT , 
  \quad f \in C^\infty_0(T^2_0(dS_D)) .
}
The most general such $f^{\mu\nu}(x)$
may be uniquely decomposed into the parts \cite{Stewart:1990aa}
\eq{
  f_{\mu\nu}(x) = f^{\rm TT}_{\mu\nu}(x) + \nabla_{(\mu} \xi_{\nu)}(x)
  + g_{\mu\nu} f_1(x) + \nabla_\mu\nabla_\nu f_2(x) ,
}
where $f^{\rm TT}_{\mu\nu}(x)$ and transverse traceless and $\xi_\nu(x)$
is transverse. Due to the self-adjointness of $L^{(1)}_{\mu\nu}$ neither 
the $\xi^\mu(x)$ nor $f_2(x)$ term contribute to (\ref{eq:TEOM}). 
Similarly, using (\ref{eq:dDEOM1}), the self-adjointness of $\P^{(0)}_{\mu\nu}$, 
and the fact that $\P^{(0)\mu\nu} f^{\rm TT}_{\mu\nu}(x) = 0$ it follows that 
$f^{\rm TT}_{\mu\nu}(x)$ does not contribute to (\ref{eq:TEOM}). Therefore
(\ref{eq:TEOM}) is non-trivial only for test functions of the form
$f^{\mu\nu}(x) = f_1(x) g^{\mu\nu}$.
Noting that
\eq{
  g^{\mu\nu} \P^{(0)}_{\mu\nu} \phi(x) = (1-D)(\Box+D) \phi(x),
}
we obtain the equation
\eq{ \label{eq:dDEOM2}
  \frac{(1-D)(D-2)}{2} b_\beta   \int_\bx p_{\bmu\bnu}(\bx) \P^{(\beta)\bmu\bnu}
  \left[ (\Box+D)(\Box- M^2_\beta) B(Z-i\epsilon s) \right]  = 0 .
}
It is important to note that the $\P^{(\beta)}_{\mu\nu} Z = 0$ so
the term in brackets is defined only up to a term linear in $Z$.
Finally, noting that
\eq{
  \P_{\bmu\bnu}^{(\beta)} = b_\beta g_{\bmu\bnu}(\bBox-M^2_\beta) + \grad,
}
(\ref{eq:dDEOM2}) implies the following equation of
motion for $B(Z-i\epsilon s)$:
\eq{ \label{eq:BEOM}
  \frac{(1-D)(D-2)}{2} \frac{(\beta-2)^2}{(2-\beta D)^2}
  (\Box+D) \left(\Box - M^2_\beta \right)\left(\bBox - M^2_\beta \right) 
  B(Z-i\epsilon s) = c Z  ,
}
where $c$ is an arbitrary constant. We have inserted the value of $b_\b$
from (\ref{eq:b}). To this point our analysis agrees with that of
\cite{Mora:2012zi}, except that these authors only considered the
case of $c=0$.

There exists a maximally-symmetric solution to (\ref{eq:BEOM}) with 
$c = - (D+1)/\vol$. We may construct this solution
using the same K\"allen-Lehmann technique used in \S\ref{sec:TT}:
defining $\s_\b$ via $M^2_\b = -\s_\b(\s_\b+D-1)$, the solution
is simply
\eqn{ \label{eq:B}
  B(Z) &=& \frac{2}{(1-D)(D-2)}\frac{(2-\beta D)^2}{(\beta-2)^2}
  \nn \\ & &
  \int_{C_B} \frac{d\omega}{2\pi i} 
  \frac{(2\omega+D-1)}{(\omega-1)(\omega+D)(\omega-\s_\b)^2
  (\omega+\s_\b+D-1)^2} \D_\omega(Z), \quad
}
with integration contour $C_B$ depicted in Fig.~\ref{fig:gravitonContours}, 
i.e., it is traversed from $-i\infty$ to $+i\infty$ mostly in the left 
half-plane but is deformed so as to keep the poles at $\omega=0,2,3,4,\dots$ 
to the right and the poles at $\omega = 1, -D, \s_\b, -(\s_\b+D-1)$ to 
the left. 
The contour $C_B$ exists so long as the
left and right poles do not overlap, i.e. so long as
\eq{ \label{eq:badBeta}
  M_\beta^2 = \frac{2(D-1)}{\beta-2} = -\s_\b(\s_\b+D-1) \neq - n(n+D-1), \quad
  {\rm for}\; n = 0, 2,3,4,\dots.
}
It is easy to verify that (\ref{eq:B}) solves (\ref{eq:BEOM})
with $c = -(D+1)/\vol$.
To summarize, excepting the values of $\beta$ listed in (\ref{eq:badBeta}), 
the graviton 2-pt function for the state $\Omega$ in dD gauge 
is manifestly maximally-symmetric and is given by (\ref{eq:dDansatz1}) 
with
$\Delta^{{\rm dD}\bmu\bnu}_{\mu\nu}(Z-i\epsilon s)$ as in (\ref{eq:dDansatz2}) and
$B(Z-i\epsilon s)$ as in (\ref{eq:B}).
Just as in TT gauge, the time-ordered 2-pt function 
$\CO{T h_{\mu\nu}(x) h^{\bmu\bnu}(\bx)}$
is provided from the Wightman-ordered 2-pt function simply by changing the 
$i\epsilon$ prescription $i\epsilon s \to i \epsilon$.

It is important to note that $\Delta_{\mu\nu}^{{\rm dD}\bmu\bnu}(Z-i\epsilon s)$ is
``pure gauge'' and as a result the addition of
$\Delta_{\mu\nu}^{{\rm dD}\bmu\bnu}(Z-i\epsilon s)$ to the 2-pt function
does not alter the positive or Hadamard attributes of $\Omega$.
Recall that because 
$\Delta_{\mu\nu}^{{\rm dD}\bmu\bnu}(Z-i\epsilon s)$ 
solves the linearized equations of motion we may use these equations 
to restrict the class of test functions used to construct $h(f)$ correlators
to $\testTT$; however, $\Delta_{\mu\nu}^{{\rm dD}\bmu\bnu}(Z-i\epsilon s)$ 
vanishes against any $f \in \testTT$ because it is composed solely
of gradient and metric terms. Likewise
$\Delta_{\mu\nu}^{{\rm dD}\bmu\bnu}(Z-i\epsilon s)$ does not contribute 
to correlators of the linearized Weyl tensor.

\subsection{Relation to previous works}
\label{sec:compare2}

We conclude by once again contrasting our results with the conflicting
claims in the literature. We begin by comparing our dD 2-pt functions
with those which may be inferred, for $D=4$, from the ``covariant 
gauge'' graviton 2-pt function of 
\cite{Higuchi:2001uv,Higuchi:2000ge,Faizal:2011iv}. This procedure
yields precisely the same dD 2-pt functions as our own. 
In this case as well one finds that maximally-symmetric 2-pt functions
exist for all but the discrete set of values for $\beta$ (\ref{eq:badBeta}).

On the other hand, our results are in conflict with the claims
of \cite{Miao:2011fc,Mora:2012zi}.
These works contend that there does not exist a 
maximally-symmetric 2-pt function for gauge parameter $\beta < 2$, including 
transverse gauge $\beta = 0$.
The source of this tension is essentially the same as that described in
\S\ref{sec:compare}: the authors of works impose the
additional requirement that $B(Z-i\epsilon s)$ admit a fourier transform in
Poincar\'e coordinates which is convergent about $\vec{k} = 0$.
This requirement can only be satisfied if one lets $c=0$ in (\ref{eq:BEOM}). 
In contrast, we have shown that the maximally-symmetric solution 
to (\ref{eq:BEOM}) corresponds to $c = - (D+1)/\vol$.
As result of the authors' preferences, \cite{Miao:2011fc,Mora:2012zi} 
were forced to consider solutions (\ref{eq:BEOM}) which are less symmetric. \\

{\bf Acknowledgements:}\\

We thank Atsushi Higuchi, Stefan Hollands, David Hunt, and Chris Fewster 
for useful conversations. We are especially thankful to Markus Fr\"ob
and Richard Woodard for several correspondences.
IM is supported by the Simons Foundation Postdoctoral
Fellowship.

\appendix

\section{The Euclidean TT 2-pt function}
\label{app:Euclidean}

In this appendix we compute the graviton 2-pt function in TT gauge
on the Euclidean sphere $S^D$ using standard harmonic analysis.
We then show that the analytic continuation of this 2-pt function
to Lorentz signature agrees with the 2-pt functions obtained in
\S\ref{sec:TT}. Our Euclidean analysis is very similar to that of 
\cite{Allen:1986tt,Floratos:1987aa,Higuchi:2001uv,Higuchi:2000ge,Park:2008ki}.
We set the radius of $S^D$ to unity.

%
%

We denote scalar harmonics on $S^D$ by $Y_\vL(x)$; these harmonics are
eigenfunctions of the Laplace operator
\eq{ \label{eq:scalarE}
  \Box Y_\vL(x) = - L(L+D-1) Y_\vL(x) ,
}
and are labelled by their angular momenta 
$\vL = (L:=L_{D},L_{D-1},\dots,L_1)$ which satisfy
$L_{D} \ge L_{D-1} \ge \dots \ge L_2 \ge |L_1| \ge 0$.
They form an orthonormal and complete set with respect to the scalar
$L^2$ inner product:
\eqn{
  \int_x Y_\vL(x) Y_\vM^*(x) &=& \delta_{\vL\vM}, \quad
  \sum_\vL Y_\vL(x) Y_\vL^*(\bx) = \frac{\delta^D(x,\bx)}{\sqrt{g}}  .
}
In this Appendix $\int_x F(x) = \int d^Dx \sqrt{\gamma(x)} F(x)$ where
$\gamma_{\mu\nu}$ is the metric on $S^D$.
Similarly, we denote symmetric, transverse, traceless rank-2 tensor harmonics
by $T^{(\vL;\a)}_{\mu\nu}(x)$; these are also eigenfunctions of the
Laplace operator
\eq{ \label{eq:tensorE}
  \Box T^{(\vL;\a)}_{\mu\nu}(x) 
  = \left[-L(L+D-1)+2\right] T^{(\vL;\a)}_{\mu\nu}(x) ,
}
and are labelled by their
angular momenta $\vL$ which satisfy
$L_{D} \ge L_{D-1} \ge \dots \ge L_2 \ge |L_1| \ge 2$, as well as a 
polarization index $\alpha = 1, \dots, (D-1)(D-2)/2$.
They form an orthonormal and complete set for TT tensors with respect 
to the tensor $L^2$ inner product:
\eqn{
  \int_x 
  T^{(\vL;\a)}_{\mu\nu}(x) (T^{(\vM;\beta)\mu\nu}(x) )^*
  &=& \delta^{\vL\vM}\delta^{\a\b}, \quad
  \sum_\vL \sum_\a
  T^{(\vL;\a)}_{\mu\nu}(x) (T^{(\vL;\a)\bmu\bnu}(\bx) )^*
  =: \delta^{{\rm TT}\bmu\bnu}_{\mu\nu}(x,\bx) . \nn \\
}
The second expression defines $\delta^{{\rm TT}\bmu\bnu}_{\mu\nu}(x,\bx)$
on $S^D$.
Explicit forms of the harmonics conforming to our conventions
may be found in, e.g. \cite{Higuchi:1986wu,Marolf:2008hg}.

%
%
It is also useful to define the following maximally-symmetric
bi-tensor harmonics. For scalars these are defined by 
\eq{
  W_L(Z) := \sum_\vj Y_\vL(x) Y_\vL^*(\bx), \quad \vL = (L,\vj) ,
}
and turn out to be just a polynomial in the Euclidean embedding
distance $Z := Z(x,\bx) = \cos \theta(x,\bx)$ where $\theta(x,\bx)$
is the angular separation between points on $S^D$.
Explicitly \cite{Drummond:1975yc},
\eq{ \label{eq:scalarW}
  W_L(Z) = \frac{(2L+D-1)}{4 \pi^{(D+1)/2}} 
  \frac{\Gamma\left(\frac{D-1}{2}\right)\Gamma(L+D-1)}
  {\Gamma(D-1)\Gamma(L+1)}
  \2F1{-L}{L+D-1}{\frac{D}{2}}{\frac{1-Z}{2}} .
}
Since $L$ is an integer the hypergeometric function in this
expression reduces to a polynomial of order $L$.
Likewise for symmetric TT tensors we define
\eq{ 
  W^{L\;\;\bmu\bnu}_{\mu\nu}(Z) :=
  \sum_{\vj} \sum_{\alpha} T^{(L,\vj,\a)}_{\mu\nu}(x)
  (T^{\bmu\bnu}_{(L,\vj,\a)}(\bx))^*, \quad \vL=(L,\vj) .
}
A closed form for $W^{L\;\;\bmu\bnu}_{\mu\nu}(Z)$ may be computed 
using the method of \cite{Allen:1994yb}; the result is
\eq{
  \label{eq:tensorW}
  W^{L\;\;\bmu\bnu}_{\mu\nu}(Z) = 
  \P_{\mu\nu}^{\ph{\mu\nu}\a\b} \P^{\bmu\bnu}_{\ph{\bmu\bnu}\ba\bb}
  \left[ N_L W_L(Z) \T2{\a}{\b}{\ba}{\bb}\right],
}
with
\eq{
  N_L = \frac{4(D-3)^2}{(D-2)^2}\frac{1}{L^2(L+D-1)^2(L+1)^2(L+D)^2} .
}
Clearly $W_L(Z)$ and $W^{L\;\;\bmu\bnu}_{\mu\nu}(Z)$ are bi-eigenfunctions 
of the Laplacian with eigenvalues as in (\ref{eq:scalarE}) and 
(\ref{eq:tensorE}) respectively.

%
%
In TT gauge the graviton 2-pt function on $S^D$ satisfies
\eq{
  \half (\Box -2)\C{h_{\mu\nu}(x) h^{\bmu\bnu}(\bx)}_{\rm E} =
  - \delta^{{\rm TT}\bmu\bnu}_{\mu\nu}(x,\bx) .
}
It is easy to invert $-\half(\Box-2)$ on the TT identity operator
and simplify:
\eqn{ \label{eq:hhE}
  \C{h_{\mu\nu}(x) h^{\bmu\bnu}(\bx)}_{\rm E} 
  &=& 2 \sum_{\vL} \sum_{\alpha} \frac{T^{(\vL,\a)}_{\mu\nu}(x)
  T^{*\bmu\bnu}_{(\vL,\a)}(\bx)}{L(L+D-1)}
  \nn \\
  &=& 2 \sum_{L=2}^\infty \frac{W^{L\;\;\bmu\bnu}_{\mu\nu}(Z)}{L(L+D-1)}
  \nn \\
  &=& 
  \P_{\mu\nu}^{\ph{\mu\nu}\a\b} \P^{\bmu\bnu}_{\ph{\bmu\bnu}\ba\bb}
  \left[ F(Z) \T2{\a}{\b}{\ba}{\bb}\right] ,
}
with
\eq{ \label{eq:Fsum}
  F(Z) = \frac{8(D-3)^2}{(D-2)^2} 
  \sum_{L=2}^\infty \frac{W_L(Z)}{L^3(L+D-1)^3(L+1)^2 (L+D)^2} .
}
The sum of polynomials of $Z$ in (\ref{eq:Fsum}) converges for non-coincident
points on $S^D$, i.e. for $Z \in [-1,1)$; however, it does not converge
for generic $Z$ values outside this range. Since the analytic continuation
of (\ref{eq:hhE}) amounts to extending the range of $Z$ from $Z \in [-1,1)$
to $Z \in \mathbb{C}\setminus [1,+\infty)$ we must first render $F(Z)$ 
into a more suitable form before the continuation.
This may be done by a standard technique known as a Watson-Sommerfeld
transformation \cite{Marolf:2010zp} whereby the sum is recast as a
contour integral. Consider the expression
\eq{ \label{eq:Fcontour}
  F(Z) = - \frac{8(D-3)^2}{(D-2)^2}  \int_{C_F} \frac{d\omega}{2\pi i}
  \frac{\pi}{\sin(\pi \omega)}
  \frac{W_\omega(-Z)}{\omega^3(\omega+D-1)^3(\omega+1)^2 (\omega+D)^2} ,
}
where the integration contour $C_F$ encloses in a clockwise fashion
the poles in the integrand at $\omega = 2,3,4,\dots$
due to the factor of $\sin(\pi\omega)$ in the denominator.
Using Cauchy's formula the contour integral is equivalent to the sum of
residues due to these poles, and this sum is precisely 
(\ref{eq:Fsum}). On the other hand, we may rewrite (\ref{eq:Fcontour})
by noting first that
\eq{
  \D_\omega(Z) = - \frac{\pi}{(2\omega+D-1) \sin\pi\omega} W_\omega(-Z) ,
}
and second that the integrand decays sufficiently rapidly as 
$|\omega| \to \infty$ such that the contour may be deformed from $C_F$ 
to $C_A$ -- see Fig.~\ref{fig:gravitonContours}. 
After these manipulations the contour integral for $F(Z)$
is precisely that for $A(Z)$ in (\ref{eq:Acontour}).
The analytic continuation process is completed by adding the appropriate
$i\epsilon$ prescription for avoiding the cut in $F(Z)=A(Z)$ along
$Z\in [1,+\infty)$. Thus the analytic continuation of the graviton 
2-pt function in TT gauge on $S^D$ agrees with our result constructed 
explicitly on $dS_D$ in \S\ref{sec:TT}.


\addcontentsline{toc}{section}{References}
\bibliographystyle{utphys}
\bibliography{refs}

\providecommand{\href}[2]{#2}\begingroup\raggedright\begin{thebibliography}{10}

\bibitem{Weinberg:2008zzc}
S.~Weinberg, {\em {Cosmology}}.
\newblock Oxford University Press, Oxford, UK, 2008.
\newblock
616 p.

\bibitem{Tsamis:1994ca}
N.~C. Tsamis and R.~P. Woodard, ``{Strong infrared effects in quantum
  gravity},''
\href{http://dx.doi.org/10.1006/aphy.1995.1015}{{\em Ann. Phys.} {\bfseries
  238} (1995) 1--82}.

\bibitem{Giddings:2010nc}
S.~B. Giddings and M.~S. Sloth, ``{Semiclassical relations and IR effects in de
  Sitter and slow-roll space-times},''
  \href{http://dx.doi.org/1475-7516/2011/i=01/a=023}{{\em JCAP} {\bfseries
  2011} no.~01, (2011) 023},
\href{http://arxiv.org/abs/1005.1056}{{\ttfamily arXiv:1005.1056 [hep-th]}}.

\bibitem{Giddings:2011zd}
S.~B. Giddings and M.~S. Sloth, ``{Cosmological observables, IR growth of
  fluctuations, and scale-dependent anisotropies},''
  \href{http://dx.doi.org/10.1103/PhysRevD.84.063528}{{\em Phys. Rev.}
  {\bfseries D84} (2011) 063528},
\href{http://arxiv.org/abs/1104.0002}{{\ttfamily arXiv:1104.0002 [hep-th]}}.

\bibitem{Giddings:2011ze}
S.~B. Giddings and M.~S. Sloth, ``{Fluctuating geometries, q-observables, and
  infrared growth in inflationary spacetimes},''
  \href{http://dx.doi.org/10.1103/PhysRevD.86.083538}{{\em Phys. Rev.}
  {\bfseries D86} (2012) 083538},
\href{http://arxiv.org/abs/1109.1000}{{\ttfamily arXiv:1109.1000 [hep-th]}}.

\bibitem{Kitamoto:2012ep}
H.~Kitamoto and Y.~Kitazawa, ``{Soft Gravitons Screen Couplings in de Sitter
  Space},''
\href{http://arxiv.org/abs/1203.0391}{{\ttfamily arXiv:1203.0391 [hep-th]}}.

\bibitem{Tsamis:1996qm}
N.~C. Tsamis and R.~P. Woodard, ``{The quantum gravitational back-reaction on
  inflation},'' \href{http://dx.doi.org/10.1006/aphy.1997.5613}{{\em Annals
  Phys.} {\bfseries 253} (1997) 1--54},
\href{http://arxiv.org/abs/hep-ph/9602316}{{\ttfamily arXiv:hep-ph/9602316}}.

\bibitem{Tsamis:1996qq}
N.~C. Tsamis and R.~P. Woodard, ``{Quantum Gravity Slows Inflation},''
  \href{http://dx.doi.org/10.1016/0550-3213(96)00246-5}{{\em Nucl. Phys.}
  {\bfseries B474} (1996) 235--248},
\href{http://arxiv.org/abs/hep-ph/9602315}{{\ttfamily arXiv:hep-ph/9602315}}.

\bibitem{Mottola:2010gp}
E.~Mottola, ``{New Horizons in Gravity: The Trace Anomaly, Dark Energy and
  Condensate Stars},''
\href{http://arxiv.org/abs/1008.5006}{{\ttfamily arXiv:1008.5006 [gr-qc]}}.

\bibitem{Tsamis:2011uq}
N.~C. Tsamis and R.~P. Woodard, ``A gravitational mechanism for cosmological
  screening,'' \href{http://dx.doi.org/10.1142/S0218271811020652}{{\em Int. J.
  Mod. Phys.} {\bfseries D20} (03, 2011) 2847--2851},
  \href{http://arxiv.org/abs/1103.5134v1}{{\ttfamily arXiv:1103.5134v1
  [gr-qc]}}.

\bibitem{Mottola:1985qt}
E.~Mottola, ``{Thermodynamic instability of de Sitter space},''
\href{http://dx.doi.org/10.1103/PhysRevD.33.1616}{{\em Phys. Rev.} {\bfseries
  D33} (1986) 1616--1621}.

\bibitem{Mazur:1986et}
P.~Mazur and E.~Mottola, ``{Spontaneous breaking of de Sitter symmetry by
  radiative effects},''
\href{http://dx.doi.org/10.1016/0550-3213(86)90058-1}{{\em Nucl. Phys.}
  {\bfseries B278} (1986) 694}.

\bibitem{Floratos:1987aa}
E.~Floratos, J.~Iliopoulos, and T.~Tomaras, ``Tree-level scattering amplitudes
  in de sitter space diverge,''
  \href{http://dx.doi.org/10.1016/0370-2693(87)90403-5}{{\em Phys. Lett.}
  {\bfseries B197} no.~3, (1987) 373 -- 378}.

\bibitem{Hollands:2008vx}
S.~Hollands and R.~M. Wald, ``{Axiomatic quantum field theory in curved
  spacetime},'' \href{http://dx.doi.org/10.1007/s00220-009-0880-7}{{\em Comm.
  Math. Phys.} {\bfseries 293} (2010) 85--125},
\href{http://arxiv.org/abs/0803.2003}{{\ttfamily arXiv:0803.2003 [gr-qc]}}.

\bibitem{Marolf:2010nz}
D.~Marolf and I.~A. Morrison, ``{Infrared stability of de Sitter QFT: Results
  at all orders},'' \href{http://dx.doi.org/10.1103/PhysRevD.84.044040}{{\em
  Phys. Rev.} {\bfseries D84} no.~4, (2011) 044040},
\href{http://arxiv.org/abs/1010.5327}{{\ttfamily arXiv:1010.5327 [gr-qc]}}.

\bibitem{Hollands:2010pr}
S.~Hollands, ``{Correlators, Feynman diagrams, and quantum no-hair in deSitter
  spacetime},''
\href{http://arxiv.org/abs/1010.5367}{{\ttfamily arXiv:1010.5367 [gr-qc]}}.

\bibitem{Higuchi:2010aa}
A.~Higuchi, D.~Marolf, and I.~A. Morrison, ``{On the Equivalence between
  Euclidean and In-In formalisms in de Sitter QFT},''
  \href{http://dx.doi.org/10.1103/PhysRevD.83.084029}{{\em Phys. Rev. D.}
  {\bfseries 83} no.~8, (Apr, 2011) 084029},
  \href{http://arxiv.org/abs/1012.3415}{{\ttfamily arXiv:1012.3415 [gr-qc]}}.

\bibitem{Korai:2012fi}
Y.~Korai and T.~Tanaka, ``{QFT in the flat chart of de Sitter space},''
\href{http://arxiv.org/abs/1210.6544}{{\ttfamily arXiv:1210.6544 [gr-qc]}}.

\bibitem{Friedrich:1986aa}
H.~Friedrich, ``Existence and structure of past asymptotically simple solutions
  of einstein's field equations with positive cosmological constant,''
  \href{http://dx.doi.org/DOI: 10.1016/0393-0440(86)90004-5}{{\em Journal of
  Geometry and Physics} {\bfseries 3} no.~1, (1986) 101 -- 117}.

\bibitem{Anderson:2004ir}
M.~T. Anderson, ``{Existence and stability of even dimensional asymptotically
  de Sitter spaces},'' \href{http://dx.doi.org/10.1007/s00023-005-0224-x}{{\em
  Annales Henri Poincare} {\bfseries 6} (2005) 801--820},
\href{http://arxiv.org/abs/gr-qc/0408072}{{\ttfamily arXiv:gr-qc/0408072}}.

\bibitem{Wald:1983ky}
R.~M. Wald, ``{Asymptotic behavior of homogeneous cosmological models in the
  presence of a positive cosmological constant},''
\href{http://dx.doi.org/10.1103/PhysRevD.28.2118}{{\em Phys.Rev.} {\bfseries
  D28} (1983) 2118--2120}.

\bibitem{Frob:2013ht}
M.~B. Fr\"ob, D.~B. Papadopoulos, A.~Roura, and E.~Verdaguer,
  ``{Nonperturbative semiclassical stability of de Sitter spacetime for small
  metric deviations},''
\href{http://arxiv.org/abs/1301.5261}{{\ttfamily arXiv:1301.5261 [gr-qc]}}.

\bibitem{Woodard:2004ut}
R.~P. Woodard, ``{de Sitter breaking in field theory},''
\href{http://arxiv.org/abs/gr-qc/0408002}{{\ttfamily arXiv:gr-qc/0408002}}.

\bibitem{Miao:2010vs}
S.~P. Miao, N.~C. Tsamis, and R.~P. Woodard, ``{de Sitter breaking through
  infrared divergences},'' \href{http://dx.doi.org/10.1063/1.3448926}{{\em J.
  Math. Phys.} {\bfseries 51} no.~7, (2010) 072503},
\href{http://arxiv.org/abs/1002.4037}{{\ttfamily arXiv:1002.4037 [gr-qc]}}.

\bibitem{Higuchi:2001uv}
A.~Higuchi and S.~S. Kouris, ``{The covariant graviton propagator in de Sitter
  spacetime},'' \href{http://dx.doi.org/10.1088/0264-9381/18/20/311}{{\em
  Class. Quant. Grav.} {\bfseries 18} (2001) 4317--4328},
\href{http://arxiv.org/abs/gr-qc/0107036}{{\ttfamily arXiv:gr-qc/0107036}}.

\bibitem{Higuchi:2000ge}
A.~Higuchi and S.~S. Kouris, ``{On the scalar sector of the covariant graviton
  two-point function in de Sitter spacetime},''
  \href{http://dx.doi.org/10.1088/0264-9381/18/15/308}{{\em Class. Quant.
  Grav.} {\bfseries 18} (2001) 2933--2944},
\href{http://arxiv.org/abs/gr-qc/0011062}{{\ttfamily arXiv:gr-qc/0011062}}.

\bibitem{Faizal:2011iv}
M.~Faizal and A.~Higuchi, ``Physical equivalence between the covariant and
  physical graviton two-point functions in de~sitter spacetime,''
  \href{http://dx.doi.org/10.1103/PhysRevD.85.124021}{{\em Phys. Rev.}
  {\bfseries D85} (Jun, 2012) 124021},
  \href{http://arxiv.org/abs/1107.0395}{{\ttfamily arXiv:1107.0395 [gr-qc]}}.

\bibitem{Higuchi:1991tk}
A.~Higuchi, ``{Quantum linearization instabilities of de Sitter spacetime.
  1},''
\href{http://dx.doi.org/10.1088/0264-9381/8/11/010}{{\em Class. Quant. Grav.}
  {\bfseries 8} (1991) 1961--1981}.

\bibitem{Kleppe:1993fz}
G.~Kleppe, ``{Breaking of de Sitter invariance in quantum cosmological
  gravity},''
\href{http://dx.doi.org/10.1016/0370-2693(93)91000-D}{{\em Phys. Lett.}
  {\bfseries B317} (1993) 305--311}.

\bibitem{Urakawa:2010it}
Y.~Urakawa and T.~Tanaka, ``Infrared divergence does not affect the
  gauge-invariant curvature perturbation,''
  \href{http://dx.doi.org/10.1103/PhysRevD.82.121301}{{\em Phys. Rev. D}
  {\bfseries 82} no.~12, (Dec, 2010) 121301},
  \href{http://arxiv.org/abs/1007.0468}{{\ttfamily arXiv:1007.0468 [hep-th]}}.

\bibitem{Higuchi:2011aa}
A.~Higuchi, D.~Marolf, and I.~A. Morrison, ``{de Sitter invariance of the dS
  graviton vacuum},''
  \href{http://dx.doi.org/10.1088/0264-9381/28/24/245012}{{\em Class. Quant.
  Grav.} {\bfseries 28} no.~24, (2011) 5012},
  \href{http://arxiv.org/abs/1107.2712}{{\ttfamily arXiv:1107.2712 [hep-th]}}.

\bibitem{Allen:1986dd}
B.~Allen, ``The graviton propagator in homogeneous and isotropic spacetimes,''
  \href{http://dx.doi.org/10.1016/0550-3213(87)90126-X}{{\em Nucl. Phys.}
  {\bfseries B287} no.~0, (1987) 743 -- 756}.

\bibitem{Park:2008ki}
D.~S. Park, ``{Graviton and Scalar Two-Point Functions in a CDL Background for
  General Dimensions},''
  \href{http://dx.doi.org/10.1088/1126-6708/2009/06/023}{{\em JHEP} {\bfseries
  06} (2009) 023},
\href{http://arxiv.org/abs/0812.3172}{{\ttfamily arXiv:0812.3172 [hep-th]}}.

\bibitem{Miao:2009hb}
S.~P. Miao, N.~C. Tsamis, and R.~P. Woodard, ``{Transforming to Lorentz Gauge
  on de Sitter},'' \href{http://dx.doi.org/10.1063/1.3266179}{{\em J. Math.
  Phys.} {\bfseries 50} (2009) 122502},
\href{http://arxiv.org/abs/0907.4930}{{\ttfamily arXiv:0907.4930 [gr-qc]}}.

\bibitem{Miao:2011fc}
S.~P. Miao, N.~C. Tsamis, and R.~P. Woodard, ``{The Graviton Propagator in de
  Donder Gauge on de Sitter Background},''
  \href{http://dx.doi.org/10.1063/1.3664760}{{\em J. Math. Phys.} {\bfseries
  52} (2011) 122301},
\href{http://arxiv.org/abs/1106.0925}{{\ttfamily arXiv:1106.0925 [gr-qc]}}.

\bibitem{Mora:2012zi}
P.~Mora, N.~Tsamis, and R.~Woodard, ``{Graviton Propagator in a General
  Invariant Gauge on de Sitter},''
  \href{http://dx.doi.org/10.1063/1.4764882}{{\em J. Math. Phys.} {\bfseries
  53} (2012) 122502},
\href{http://arxiv.org/abs/1205.4468}{{\ttfamily arXiv:1205.4468 [gr-qc]}}.

\bibitem{Higuchi:2012vy}
A.~Higuchi, ``{Equivalence between the Weyl-tensor and gauge-invariant graviton
  two-point functions in Minkowski and de Sitter spaces},''
\href{http://arxiv.org/abs/1204.1684}{{\ttfamily arXiv:1204.1684 [gr-qc]}}.

\bibitem{Miao:2013isa}
S.~Miao, P.~Mora, N.~Tsamis, and R.~Woodard, ``{The Perils of Analytic
  Continuation},''
\href{http://arxiv.org/abs/1306.5410}{{\ttfamily arXiv:1306.5410 [gr-qc]}}.

\bibitem{Hawking:1973uf}
S.~W. Hawking and G.~F.~R. Ellis, {\em {The Large scale structure of
  space-time}}.
\newblock Cambridge University Press, Cambridge, UK, 1973.
\newblock 393 p.

\bibitem{Birrell:1982ix}
N.~D. Birrell and P.~C.~W. Davies, {\em {Quantum fields in curved space}}.
\newblock Cambridge University Press, Cambridge, UK, 1982.
\newblock 340p.

\bibitem{Spradlin:2001pw}
M.~Spradlin, A.~Strominger, and A.~Volovich, ``{Les Houches lectures on de
  Sitter space},''
\href{http://arxiv.org/abs/hep-th/0110007}{{\ttfamily arXiv:hep-th/0110007}}.

\bibitem{Allen:1985ux}
B.~Allen, ``{Vacuum States in de Sitter Space},''
\href{http://dx.doi.org/10.1103/PhysRevD.32.3136}{{\em Phys. Rev.} {\bfseries
  D32} (1985) 3136}.

\bibitem{Fewster:2012aa}
C.~J. {Fewster} and D.~S. {Hunt}, ``{Quantization of linearized gravity in
  cosmological vacuum spacetimes},''
  \href{http://arxiv.org/abs/1203.0261}{{\ttfamily arXiv:1203.0261 [math-ph]}}.

\bibitem{Wald:1995yp}
R.~M. Wald, {\em {Quantum field theory in curved space-time and black hole
  thermodynamics}}.
\newblock Univ. Pr., Chicago, USA, 1994.
\newblock 205 p.

\bibitem{Haag:1992aa}
R.~Haag, {\em Local quantum physics: fields, particles, algebras}.
\newblock Texts and monographs in physics. Springer-Verlag, 1992.
\newblock 390 p.

\bibitem{Wald:2006ty}
R.~M. Wald, ``{The History and Present Status of Quantum Field Theory in Curved
  Spacetime},''
\href{http://arxiv.org/abs/gr-qc/0608018}{{\ttfamily arXiv:gr-qc/0608018
  [gr-qc]}}.

\bibitem{Kay:2006jn}
B.~S. Kay, \href{http://dx.doi.org/10.1016/B0-12-512666-2/00083-3}{``{Quantum
  field theory in curved spacetime},''} in {\em {Encyclopedia of Mathematical
  Physics}}, J.-P. Francoise, G.~L. Naber, and T.~Tsou, eds., vol.~4,
  pp.~202--214.
\newblock Academic Press, Oxford, 2006.
\newblock
\href{http://arxiv.org/abs/gr-qc/0601008}{{\ttfamily arXiv:gr-qc/0601008
  [gr-qc]}}.
\newblock

\bibitem{Friedlander:1975aa}
F.~G. Friedlander, {\em {The wave equation on a curved space-time}}.
\newblock Cambridge Monographs on Mathematical Physics. Cambridge University
  Press, 1975.
\newblock 296 p.

\bibitem{Hollands:2001fk}
S.~Hollands and R.~M. Wald, ``Local wick polynomials and time ordered products
  of quantum fields in curved spacetime,''
  \href{http://dx.doi.org/10.1007/s002200100540}{{\em Comm. Math. Phys.}
  {\bfseries 223} no.~2, (2001) 289--326},
  \href{http://arxiv.org/abs/gr-qc/0103074}{{\ttfamily gr-qc/0103074}}.

\bibitem{Brunetti:1995rf}
R.~Brunetti, K.~Fredenhagen, and M.~Kohler, ``{The Microlocal spectrum
  condition and Wick polynomials of free fields on curved space-times},''
  \href{http://dx.doi.org/10.1007/BF02099626}{{\em Comm. Math. Phys.}
  {\bfseries 180} (1996) 633--652},
\href{http://arxiv.org/abs/gr-qc/9510056}{{\ttfamily arXiv:gr-qc/9510056
  [gr-qc]}}.

\bibitem{Mottola:1984ar}
E.~Mottola, ``{Particle Creation in de Sitter Space},''
\href{http://dx.doi.org/10.1103/PhysRevD.31.754}{{\em Phys. Rev.} {\bfseries
  D31} (1985) 754}.

\bibitem{Bousso:2001mw}
R.~Bousso, A.~Maloney, and A.~Strominger, ``{Conformal vacua and entropy in de
  Sitter space},'' \href{http://dx.doi.org/10.1103/PhysRevD.65.104039}{{\em
  Phys. Rev.} {\bfseries D65} (2002) 104039},
\href{http://arxiv.org/abs/hep-th/0112218}{{\ttfamily arXiv:hep-th/0112218}}.

\bibitem{Brunetti:2005pr}
R.~Brunetti, K.~Fredenhagen, and S.~Hollands, ``{A remark on alpha vacua for
  quantum field theories on de Sitter space},''
  \href{http://dx.doi.org/10.1088/1126-6708/2005/05/063}{{\em JHEP} {\bfseries
  05} (2005) 063},
\href{http://arxiv.org/abs/hep-th/0503022}{{\ttfamily arXiv:hep-th/0503022}}.

\bibitem{Lagogiannis:2011st}
P.~Lagogiannis, A.~Maloney, and Y.~Wang, ``{Odd-dimensional de Sitter Space is
  Transparent},''
\href{http://arxiv.org/abs/1106.2846}{{\ttfamily arXiv:1106.2846 [hep-th]}}.

\bibitem{Hormander:1990aa}
L.~H\"ormander, {\em {The Analysis of Linear Partial Differential Operators I:
  Distribution Theory and Fourier Analysis}}.
\newblock Springer-Verlag, Berlin / Heidelberg, 2nd~ed., 1990.
\newblock 440 p.

\bibitem{Radzikowski:1996aa}
M.~J. Radzikowski, ``Micro-local approach to the hadamard condition in quantum
  field theory on curved space-time,''
  \href{http://dx.doi.org/10.1007/BF02100096}{{\em Comm. Math. Phys.}
  {\bfseries 179} (1996) 529--553}.

\bibitem{Sahlmann:2001aa}
H.~Sahlmann and R.~Verch, ``Microlocal spectrum condition and hadamard form for
  vector-valued quantum fields in curved spacetime,''
  \href{http://dx.doi.org/10.1142/S0129055X01001010}{{\em Rev. Math. Phys.}
  {\bfseries 13} no.~10, (2001) 1203--1246},
  \href{http://arxiv.org/abs/0008029}{{\ttfamily arXiv:0008029 [math-ph]}}.

\bibitem{Dappiaggi:2011cj}
C.~Dappiaggi and D.~Siemssen, ``{Hadamard States for the Vector Potential on
  Asymptotically Flat Spacetimes},''
\href{http://arxiv.org/abs/1106.5575}{{\ttfamily arXiv:1106.5575 [gr-qc]}}.

\bibitem{Allen:1988aa}
B.~Allen, A.~Folacci, and A.~C. Ottewill, ``Renormalized graviton stress-energy
  tensor in curved vacuum space-times,''
  \href{http://dx.doi.org/10.1103/PhysRevD.38.1069}{{\em Phys. Rev. D}
  {\bfseries 38} (Aug, 1988) 1069--1082}.

\bibitem{DeWitt:1960aa}
B.~S. DeWitt and R.~W. Brehme, ``Radiation damping in a gravitational field,''
  \href{http://dx.doi.org/10.1016/0003-4916(60)90030-0}{{\em Annals of Physics}
  {\bfseries 9} no.~2, (1960) 220 -- 259}.

\bibitem{Higuchi:1991tn}
A.~Higuchi, ``{Linearized gravity in de Sitter space-time as a representation
  of SO(4,1)},''
\href{http://dx.doi.org/10.1088/0264-9381/8/11/011}{{\em Class. Quant. Grav.}
  {\bfseries 8} (1991) 2005--2021}.

\bibitem{Fierz:1939aa}
M.~Fierz and W.~Pauli, ``{On Relativistic Wave Equations for Particles of
  Arbitrary Spin in an Electromagnetic Field},''
  \href{http://dx.doi.org/10.1098/rspa.1939.0140}{{\em Proc. Roy. Soc. Lond. A}
  {\bfseries 173} no.~953, (1939) 211--232}.

\bibitem{Higuchi:1986py}
A.~Higuchi, ``{Forbidden mass range for spin-2 field theory in de Sitter
  spacetime},''
\href{http://dx.doi.org/10.1016/0550-3213(87)90691-2}{{\em Nucl. Phys.}
  {\bfseries B282} (1987) 397}.

\bibitem{Allen:1985wd}
B.~Allen and T.~Jacobson, ``{Vector two point functions in maximally symmetric
  spaces},''
\href{http://dx.doi.org/10.1007/BF01211169}{{\em Comm. Math. Phys.} {\bfseries
  103} (1986) 669}.

\bibitem{Allen:1986qj}
B.~Allen and C.~A. Lutken, ``{Spinor two point functions in maximally symmetric
  spaces},''
\href{http://dx.doi.org/10.1007/BF01454972}{{\em Comm. Math. Phys.} {\bfseries
  106} (1986) 201}.

\bibitem{Allen:1986tt}
B.~Allen and M.~Turyn, ``{An evaluation of the graviton propagator in de Sitter
  space},''
\href{http://dx.doi.org/10.1016/0550-3213(87)90672-9}{{\em Nucl. Phys.}
  {\bfseries B292} (1987) 813}.

\bibitem{Allen:1994yb}
B.~Allen, ``{Maximally symmetric spin two bitensors on S**3 and H**3},''
  \href{http://dx.doi.org/10.1103/PhysRevD.51.5491}{{\em Phys. Rev.} {\bfseries
  D51} (1995) 5491--5497},
\href{http://arxiv.org/abs/gr-qc/9411023}{{\ttfamily arXiv:gr-qc/9411023}}.

\bibitem{DHoker:1999aa}
E.~D'Hoker, D.~Z. Freedman, S.~D. Mathur, A.~Matusis, and L.~Rastelli,
  ``{Graviton and gauge boson propagators in AdSd+1},''
  \href{http://dx.doi.org/10.1016/S0550-3213(99)00524-6}{{\em Nucl. Phys.}
  {\bfseries B562} no.~1-2, (1999) 330 -- 352},
  \href{http://arxiv.org/abs/hep-th/9902042}{{\ttfamily hep-th/9902042
  [hep-th]}}.

\bibitem{Marolf:2010zp}
D.~Marolf and I.~A. Morrison, ``{The IR stability of de Sitter: Loop
  corrections to scalar propagators},''
  \href{http://dx.doi.org/10.1103/PhysRevD.82.105032}{{\em Phys. Rev. D}
  {\bfseries 82} no.~10, (Nov, 2010) 105032},
\href{http://arxiv.org/abs/1006.0035}{{\ttfamily arXiv:1006.0035 [gr-qc]}}.

\bibitem{Hollands:2011we}
S.~Hollands, ``{Massless interacting quantum fields in deSitter spacetime},''
  \href{http://dx.doi.org/10.1007/s00023-011-0140-1}{{\em Annales Henri
  Poincare} {\bfseries 13} (2011) 1039--1081},
\href{http://arxiv.org/abs/1105.1996}{{\ttfamily arXiv:1105.1996 [gr-qc]}}.

\bibitem{Marolf:2012kh}
D.~Marolf, I.~A. Morrison, and M.~Srednicki, ``{Perturbative S-matrix for
  massive scalar fields in global de Sitter space},''
\href{http://arxiv.org/abs/1209.6039}{{\ttfamily arXiv:1209.6039 [hep-th]}}.

\bibitem{Bateman:1955}
A.~Erdelyi, ed., {\em Higher transcendental functions}, vol.~1 of {\em Bateman
  Manuscript Project}.
\newblock McGraw-Hill, New York, 1953.

\bibitem{Allen:1987tz}
B.~Allen and A.~Folacci, ``{The massless minimally coupled scalar field in de
  Sitter space},''
\href{http://dx.doi.org/10.1103/PhysRevD.35.3771}{{\em Phys. Rev.} {\bfseries
  D35} (1987) 3771}.

\bibitem{Bros:2010aa}
J.~Bros, H.~Epstein, and U.~Moschella, ``Scalar tachyons in the de sitter
  universe,'' \href{http://dx.doi.org/10.1007/s11005-010-0406-4}{{\em Lett.
  Math. Phys.} {\bfseries 93} (2010) 203--211},
  \href{http://arxiv.org/abs/1003.1396}{{\ttfamily arXiv:1003.1396 [hep-th]}}.

\bibitem{Schlingemann:1999mk}
D.~Schlingemann, ``{Euclidean field theory on a sphere},''
\href{http://arxiv.org/abs/hep-th/9912235}{{\ttfamily arXiv:hep-th/9912235}}.

\bibitem{Kahya:2011sy}
E.~Kahya, S.~Miao, and R.~Woodard, ``{The Coincidence Limit of the Graviton
  Propagator in de Donder Gauge on de Sitter Background},''
  \href{http://dx.doi.org/10.1063/1.3681886}{{\em J. Math. Phys.} {\bfseries
  53} (2012) 022304},
\href{http://arxiv.org/abs/1112.4420}{{\ttfamily arXiv:1112.4420 [gr-qc]}}.

\bibitem{Mora:2012zh}
P.~Mora, N.~Tsamis, and R.~Woodard, ``{Weyl-Weyl Correlator in de Donder Gauge
  on de Sitter},''
\href{http://arxiv.org/abs/1205.4466}{{\ttfamily arXiv:1205.4466 [gr-qc]}}.

\bibitem{Streater:1989vi}
R.~F. Streater and A.~S. Wightman, {\em {PCT, spin and statistics, and all
  that}}.
\newblock Advanced book classics. Addison-Wesley, Redwood City, USA, 1989.
\newblock 207 p.

\bibitem{Schlieder:1965aa}
S.~Schlieder, ``{Some Remarks about the Localization of States in a Quantum
  Field Theory},'' \href{http://dx.doi.org/10.1007/BF01645904}{{\em Comm. Math.
  Phys.} {\bfseries 1} (1965) 265--280}.

\bibitem{Borchers:1998nw}
H.~J. Borchers and D.~Buchholz, ``{Global properties of vacuum states in de
  Sitter space},'' {\em Annales Poincare Phys. Theor.} {\bfseries A70} (1999)
  23--40,
\href{http://arxiv.org/abs/gr-qc/9803036}{{\ttfamily arXiv:gr-qc/9803036}}.

\bibitem{Jaekel:2000tk}
C.~D. Jaekel, ``{The Reeh-Schlieder property for ground states},''
  \href{http://dx.doi.org/10.1002/andp.200310012}{{\em Annalen Phys.}
  {\bfseries 12} (2003) 289--299},
\href{http://arxiv.org/abs/hep-th/0001154}{{\ttfamily arXiv:hep-th/0001154}}.

\bibitem{Strohmaier:2000aa}
A.~Strohmaier, ``The reeh--schlieder property for quantum fields on stationary
  spacetimes,'' \href{http://dx.doi.org/10.1007/s002200000299}{{\em Comm. Math.
  Phys.} {\bfseries 215} (2000) 105--118}.

\bibitem{Strohmaier:2002aa}
A.~Strohmaier, R.~Verch, and M.~Wollenberg, ``Microlocal analysis of quantum
  fields on curved space--times: Analytic wave front sets and reeh--schlieder
  theorems,'' \href{http://dx.doi.org/10.1063/1.1506381}{{\em J. Math. Phys.}
  {\bfseries 43} no.~11, (2002) 5514--5530},
  \href{http://arxiv.org/abs/math-ph/0202003}{{\ttfamily math-ph/0202003
  [math-ph]}}.

\bibitem{Sanders:2009aa}
K.~Sanders, ``On the reeh-schlieder property in curved spacetime,''
  \href{http://dx.doi.org/10.1007/s00220-009-0734-3}{{\em Comm. Math. Phys.}
  {\bfseries 288} (2009) 271--285}.

\bibitem{Dappiaggi:2011fk}
C.~Dappiaggi, ``Remarks on the reeh-schlieder property for higher spin free
  fields on curved spacetimes,''
  \href{http://arxiv.org/abs/1102.5270v1}{{\ttfamily 1102.5270v1}}.

\bibitem{Bros:1998ik}
J.~Bros, H.~Epstein, and U.~Moschella, ``{Analyticity properties and thermal
  effects for general quantum field theory on de Sitter space-time},''
  \href{http://dx.doi.org/10.1007/s002200050435}{{\em Comm. Math. Phys.}
  {\bfseries 196} (1998) 535--570},
\href{http://arxiv.org/abs/gr-qc/9801099}{{\ttfamily arXiv:gr-qc/9801099}}.

\bibitem{Verch:1999aa}
R.~Verch, ``Wavefront sets in algebraic quantum field theory,''
  \href{http://dx.doi.org/10.1007/s002200050680}{{\em Comm. Math. Phys.}
  {\bfseries 205} (1999) 337--367}.

\bibitem{Kouris:2001aa}
S.~S. Kouris, ``The weyl tensor two-point function in de sitter spacetime,''
  \href{http://dx.doi.org/10.1088/0264-9381/18/22/316}{{\em Class. Quant.
  Grav.} {\bfseries 18} no.~22, (2001) 4961},
  \href{http://arxiv.org/abs/gr-qc/0107064}{{\ttfamily gr-qc/0107064 [gr-qc]}}.
  A Corrigendum for this article has been published in {\it Clas. Quant. Grav.}
  {\bf 29} 169501.

\bibitem{Mora:2012kr}
P.~J. Mora and R.~P. Woodard, ``{Linearized Weyl-Weyl Correlator in a de Sitter
  Breaking Gauge},''
\href{http://arxiv.org/abs/1202.0999}{{\ttfamily arXiv:1202.0999 [gr-qc]}}.

\bibitem{Stewart:1990aa}
J.~M. Stewart, ``Perturbations of friedmann-robertson-walker cosmological
  models,'' \href{http://dx.doi.org/10.1088/0264-9381/7/7/013}{{\em Class.
  Quant. Grav.} {\bfseries 7} no.~7, (1990) 1169}.

\bibitem{Higuchi:1986wu}
A.~Higuchi, ``{Symmetric Tensor Spherical Harmonics On The N Sphere And Their
  Application To The De Sitter Group SO(N,1)},''
\href{http://dx.doi.org/10.1063/1.527513}{{\em J. Math. Phys.} {\bfseries 28}
  (1987) 1553}.

\bibitem{Marolf:2008hg}
D.~Marolf and I.~A. Morrison, ``{Group Averaging for de Sitter free fields},''
  \href{http://dx.doi.org/10.1088/0264-9381/26/23/235003}{{\em Class. Quant.
  Grav.} {\bfseries 26} no.~23, (2009) 235003},
\href{http://arxiv.org/abs/0810.5163}{{\ttfamily arXiv:0810.5163 [gr-qc]}}.

\bibitem{Drummond:1975yc}
I.~T. Drummond, ``{Dimensional Regularization of Massless Theories in Spherical
  Space-Time},''
\href{http://dx.doi.org/10.1016/0550-3213(75)90089-9}{{\em Nucl. Phys.}
  {\bfseries B94} (1975) 115}.

\end{thebibliography}\endgroup

\end{document}